\documentclass[prx,twocolumn,superscriptaddress]{revtex4-2}

\usepackage{amsmath}
\usepackage{amssymb}
\usepackage{amsfonts}
\usepackage{graphicx}
\usepackage[colorlinks=true, allcolors=blue]{hyperref}
\usepackage{braket}
\usepackage[noend]{algpseudocode}
\usepackage{algorithm}
\usepackage{qcircuit}
\usepackage{subfig}
\usepackage{color}
\usepackage{amsthm}

\theoremstyle{plain}
\newtheorem{thm}{Theorem}[section]

\newcommand{\nm}{\nonumber}
\newcommand{\bs}[1]{\boldsymbol{#1}}
\newcommand{\diff}{\mathrm{d}}
\newcommand{\tr}[1]{\mathrm{tr}\left( #1 \right)}

\newcommand{\lstickx}[1]{\lstick{\makebox[1.5em][l]{$#1$}}}
\newcommand{\arrep}[1]{\ar @<4pt> @/^/[#1]|-{\mbox{ $\times L$ }}}

\begin{document}

\title{Variational quantum algorithm for generalized eigenvalue problems\\
and its application to the finite element method}

\author{Yuki Sato}
\email{yuki-sato@mosk.tytlabs.co.jp}
\affiliation{Toyota Central R\&D Labs., Inc., 41-1, Yokomichi, Nagakute, Aichi 480-1192, Japan}
\affiliation{Quantum Computing Center, Keio University, 3-14-1 Hiyoshi, Kohoku-ku, Yokohama, Kanagawa 223-8522, Japan}

\author{Hiroshi C. Watanabe}
\affiliation{Quantum Computing Center, Keio University, 3-14-1 Hiyoshi, Kohoku-ku, Yokohama, Kanagawa 223-8522, Japan}
\affiliation{Department of Chemistry, Graduate School of Science, Kyushu University, 744 Motooka, Nishi-ku, Fukuoka, 819-0395, Japan}

\author{Rudy Raymond}
\affiliation{IBM Quantum, IBM Japan, 19-21 Nihonbashi Hakozaki-cho, Chuo-ku, Tokyo 103-8510, Japan}
\affiliation{Quantum Computing Center, Keio University, 3-14-1 Hiyoshi, Kohoku-ku, Yokohama, Kanagawa 223-8522, Japan}
\affiliation{Department of Computer Science, The University of Tokyo, 7-3-1, Hongo, Bunkyo-ku, Tokyo 113-0033, Japan}

\author{Ruho Kondo}
\affiliation{Toyota Central R\&D Labs., Inc., 41-1, Yokomichi, Nagakute, Aichi 480-1192, Japan}
\affiliation{Quantum Computing Center, Keio University, 3-14-1 Hiyoshi, Kohoku-ku, Yokohama, Kanagawa 223-8522, Japan}

\author{Kaito Wada}
\affiliation{Department of Applied Physics and Physico-Informatics, Keio University, Hiyoshi 3-14-1, Kohoku-ku, Yokohama 223-8522, Japan}

\author{Katsuhiro Endo}
\affiliation{Research Center for Computational Design of Advanced Functional Materials, National Institute of Advanced Industrial Science and Technology (AIST), 1-1-1 Umezono, Tsukuba, Ibaraki, 305-8568, Japan}
\affiliation{Quantum Computing Center, Keio University, 3-14-1 Hiyoshi, Kohoku-ku, Yokohama, Kanagawa 223-8522, Japan}

\author{Michihiko Sugawara}
\affiliation{Quantum Computing Center, Keio University, 3-14-1 Hiyoshi, Kohoku-ku, Yokohama, Kanagawa 223-8522, Japan}

\author{Naoki Yamamoto}
\affiliation{Quantum Computing Center, Keio University, 3-14-1 Hiyoshi, Kohoku-ku, Yokohama, Kanagawa 223-8522, Japan}
\affiliation{Department of Applied Physics and Physico-Informatics, Keio University, Hiyoshi 3-14-1, Kohoku-ku, Yokohama 223-8522, Japan}

\begin{abstract}
Generalized eigenvalue problems (GEPs) play an important role in the variety of fields including engineering, machine learning and quantum chemistry.
Especially, many problems in these fields can be reduced to finding the minimum or maximum eigenvalue of GEPs.
One of the key problems to handle GEPs is that the memory usage and computational complexity explode as the size of the system of interest grows.
This paper aims at extending sequential quantum optimizers for GEPs.
Sequential quantum optimizers are a family of algorithms that iteratively solve the analytical optimization of single-qubit gates in a coordinate descent manner.
The contribution of this paper is as follows.
First, we formulate the GEP as the minimization/maximization problem of the fractional form of the expectations of two Hermitians. 
We then showed that the fractional objective function can be analytically minimized or maximized with respect to a single-qubit gate by solving a GEP of a 4 × 4 matrix.
Second, we show that a system of linear equations (SLE) characterized by a positive-definite Hermitian can be formulated as a GEP and thus be attacked using the proposed method. 
Finally, we demonstrate two applications to important engineering problems formulated with the finite element method.
Through the demonstration, we have the following bonus finding; a problem having a real-valued solution can be solved more effectively using quantum gates generating a complex-valued state vector, which demonstrates the effectiveness of the proposed method.
\end{abstract}

\maketitle

\section{Introduction}

Generalized eigenvalue problems (GEPs) are expressed as 
\begin{equation}
\label{eq:g_eig}
    A\bs{v} = \lambda B \bs{v},
\end{equation}
where $A$ and $B$ are Hermitian matrices; also $\lambda$ and $\bs{v}$ are the generalized eigenvalue and generalized eigenvector, respectively.
GEPs play an important role in the variety of fields, including engineering~\cite{boffi2010finite}, machine learning~\cite{yu2011kernel} and quantum chemistry~\cite{ford1974generalized}.
In the field of engineering, finding the lowest eigenvalue of a symmetric generalized eigenvalue problem often appears in the finite element approximation of mechanical structures to estimate their dynamical properties~\cite{hughes2012finite}.
Many problems in machine learning can be reduced to finding the minimum or maximum eigenvalue of (generalized) eigenvalue problems, such as for the principal component analysis, canonical correlation analysis, and Fisher discriminant analysis~\cite{yu2011kernel}.
A key problem is that, as the size of the system grows, memory usage and computational complexity explode.
Actually, many works have been performed using supercomputers to deal with the system with several tens of billion degrees of freedom~\cite{Klawonn2015, Toivanen2018}.

Quantum computing is a promising and attractive approach to realize high performance computing that is significantly faster than classical computing thanks to the capability of handling an exponentially large Hilbert space. 
For fault-tolerant quantum computers, the quantum phase estimation algorithm, which can be used to calculate the ground state of a system Hamiltonian in quantum chemistry~\cite{aspuru2005simulated, o2016scalable}, has already been applied to generalized eigenvalue problems~\cite{parker2020quantum}.
Meanwhile, for near-term quantum computers, the variational quantum eigensolver (VQE)~\cite{peruzzo2014NatCom, kandala2017Nat}, 
which calculates the minimal eigenvalue of a Hamiltonian based on a classical-quantum 
hybrid scheme, has been extensively studied, especially for quantum chemistry~\cite{li2019variational, zhang2021shallow}.
VQE is a kind of the variational quantum algorithm (VQA)~\cite{cerezo2021variational,TILLY20221}, where a certain function of expectation values of observables is minimized or maximized through a parameterized quantum circuit (PQC) or simply an \textit{ansatz}.
VQAs have been applied to various problems, including ground state calculations~\cite{peruzzo2014NatCom, kandala2017Nat, gomes2021adaptive, zhang2021shallow}, excited state calculations~\cite{higgott2019variational, nakanishi2019subspace, gocho2023excited, hirai2023excited}, time evolution simulations~\cite{benedetti2021, wada2022} for chemical calculations, partial differential equation solvers~\cite{liu2021variational, sato2021variational, demirdjian2022variational}, algebraic operations such as linear system solvers~\cite{bravo2019variational, xu2021variational} and principle component analysis~\cite{larose2019variational, cerezo2022variational}.
VQAs have also been applied to generalized eigenvalue problems ~\cite{liang2020variational, liang2022quantum}, where the cost function derived from the generalized Rayleigh quotient is optimized based on the gradient-based optimizer.

The sequential quantum optimizers are a family of algorithms which iterate the analytical optimizations of single-qubit gates in a coordinate descent manner.
This concept was first studied in Ref.~\cite{nakanishi2020}, proposing the sequential optimizer of single-qubit gates in the PQC, particularly their rotation angles; we call this method the Nakanishi-Fujii-Todo (NFT) algorithm. 
Rotosolve~\cite{ostaszewski2021}, which was proposed independently from NFT, also optimizes the angle of a single-qubits gate, and Rotoselect~\cite{ostaszewski2021}, which was proposed together with Rotosolve, optimizes the rotational angles of single-qubit gates selecting the optimal axes from a finite discrete set of axes. 
These methods were extended to the continuous optimization of rotational axes, which are called the ``Free-axis selection'' (\textit{Fraxis})~\cite{watanabe2021} and the maximum optimization of a single-qubit gate termed ``Free-quaternion selection''(\textit{FQS})~\cite{wada2022full}.
The methods exhibited better convergences compared with the gradient-based approaches~\cite{nakanishi2020, wada2022full} and a recent finding suggested that their behaviors with regards to the so-called {\it barren plateaus} are
similar to their gradient-based counterparts~\cite{wada2022full}.
These methods are applicable to several problems such as general VQE and simulation for real/imaginary time evolutions, but they are limited to the case where the objective function is in the form of the expectation of a problem Hamiltonian. 
In our case focusing on the GEPs in Eq.~\eqref{eq:g_eig}, unfortunately, $A$ must be identity; that is, the GEPs could not be solved in the same way as the aforementioned methods.

In this paper, we extend the VQA based on the sequential quantum optimizers, in particular focusing on FQS~\cite{wada2022full}, to solve GEPs characterized by two symmetric (more generally Hermitian) matrices. 

The contributions of this paper are as follows.
First, we reformulate the GEPs by the minimization/maximization problem of the generalized Rayleigh quotient, which is further reformulated in the fractional form of the expectations of two Hamiltonians as the objective function of the VQA.
We then show that the sequential quantum optimization method is applicable; that is, the objective function can be analytically minimized or maximized with respect to a single-qubit gate by solving a GEP of a $4 \times 4$ matrix.

Second, we deal with the general problem of solving a system of linear equations (SLE), where an efficient solver is in great demand e.g., for solving partial differential equations (PDE)~\cite{evans2010partial}, and in machine learning~\cite{deisenroth2020mathematics,aggarwal2020linear}.
We show that the SLE problem, which is characterized by a positive-definite Hermitian, can be formulated in a form of GEP and thus be attacked using the proposed variational approach mentioned above. 
While several studies~\cite{liu2021variational, sato2021variational, liu2022application} have already applied VQAs to solving an SLE derived from PDEs, the proposed method is advantageous with respect to less controlled-unitary gates and auxiliary qubits required owing to expressing an SLE as a GEP.

Finally, we demonstrate two applications to engineering problems of importance formulated with the finite element method~\cite{allaire2007numerical, hughes2012finite}. 
One of the problems is for solving an SLE derived from a Poisson equation, and the other is eigen-frequency analysis of a beam structure. 
Based on these results, we give an estimate that a few dozen of qubits are required to solve practical problems. 
Through the demonstration, we have the following bonus finding; a problem having a real-valued solution can be solved more effectively using quantum gates generating a complex-valued state vector.

The rest of this paper is organized as follows.
In Sec.~\ref{sec:preliminary}, we briefly discuss the GEP, which can be solved by minimizing/maximizing the generalized Rayleigh quotient.
We also give an overview of sequential quantum optimizer, which optimizes a PQC in a coordinate descent manner. 
In Sec.~\ref{sec:method}, we construct the method to fully minimize the generalized Rayleigh quotient with respect to a single-qubit gate by extending the FQS.
We also give a VQA for solving an SLE in the GEP formulation. 
Section.~\ref{sec:results} shows the two demonstrations. 
Finally, we conclude this study in Sec.~\ref{sec:conclusions}. 

\section{Preliminaries} \label{sec:preliminary}
\subsection{Generalized eigenvalue problem}
In this paper, we focus on the minimum eigenvalue of the generalized eigenvalue 
problem (GEP) in Eq.~\eqref{eq:g_eig}, assuming that $B \in \mathbb{C}^{N \times N}$ 
is a positive-definite Hermitian matrix and $A \in \mathbb{C}^{N \times N}$ is a Hermitian 
matrix.
To this end, it is convenient to introduce the generalized Rayleigh quotient $R$ defined as
\begin{align} \label{eq:g_rayleigh}
R(\bs{w}; A, B) := \dfrac{\bs{w}^\dagger A \bs{w}}{\bs{w}^\dagger B \bs{w}},
\end{align}
where $\bs{w}$ is an arbitrary unit vector in $\mathbb{C}^N$.
The minimum eigenvalue of the GEP \eqref{eq:g_eig} is identical to the minimum 
value of the generalized Rayleigh quotient $R$. 
Also, the minimizer $\bs{w}$ of $R$ is identical to the eigenvector $\bs{v}$ 
corresponding to the minimum eigenvalue as explained in Appendix~\ref{sec:rayleigh}.
When one is interested in the maximum eigenvalue, it is enough to replace the minimum with the maximum over the following discussion.

\subsection{Overview of sequential quantum optimizers}

As an optimizer of PQC, this study employs a coordinate-descent sequential optimizer which sequentially optimizes single-qubit gates in a PQC. 
The reasons to employ the sequential optimizers, in addition to their analytically computable solutions, are their better convergences~\cite{nakanishi2020,TILLY20221,wada2022full,watanabe2021} and a recent finding that their behaviors with regards to the so-called Barren Plateaus are similar to their gradient-based counterparts~\cite{wada2022full}.
Here, we give an overview of sequential quantum optimizers, especially FQS~\cite{wada2022full}.

Let $\rho$ denote a quantum state prepared through a PQC from an initial state $\rho_0$ as follows:
\begin{align}
\rho = U_D \cdots U_d \cdots U_1 \rho_0 U_1^\dagger \cdots U_d^\dagger \cdots U_D^\dagger,
\end{align}
where $D$ is the number of parameterized single-qubit gates and $U_d~(d=1, \ldots, D)$ 
is the $d$-th parameterized single-qubit gate. 
We herein omit representing fixed unitary gates that include non-local gates.
Now, we represent the $d$-th parameterized single-qubit gate $U_d$ as
\begin{align}
U_d := \bs{q}_d \cdot \Vec{\varsigma},
\end{align}
where
\begin{align}
    \bs{q}_d = \begin{pmatrix} 
    \cos \left( \dfrac{\theta_d}{2} \right) \\
    \sin \left( \dfrac{\theta_d}{2} \right) n_{dx} \\
    \sin \left( \dfrac{\theta_d}{2} \right) n_{dy} \\
    \sin \left( \dfrac{\theta_d}{2} \right) n_{dz} \end{pmatrix}, 
\label{single-para-vector}
\end{align}
and $\vec{\varsigma}=(I, -\iota X, -\iota Y, -\iota Z)$ with $\iota$ the imaginary unit and $X$, $Y$, and $Z$ the Pauli matrix.
$\theta_d$ and $\bs{n}_d = (n_{dx}, n_{dy}, n_{dz})$ are the rotational angle and axis of the $d$-th gate, respectively.
That is, $U_d$ is parameterized by the unit quaternion $\bs{q}_d$. 
Accordingly, the quantum state is parametrized with the set of 
$\{\bs{q}_d\}_{d=1}^D$ as $\rho=\rho(\{\bs{q}_d\}_{d=1}^D)$. 
The strategy of sequential optimization is to repeat the exact optimization of 
$\bs{q}_d$ for all $d=1, \ldots, D$. 
Actually, the expectation $\braket{H}$ can be expressed as the quadratic form of the 
unit quaternion $\bs{q}_d$ as follows:
\begin{align} \label{eq:quadratic_form}
 \braket{H}(\bs{q}_d) = \bs{q}_d^\top S(\rho', H') \bs{q}_d,
\end{align}
where 
\begin{align}
    \rho' &= U_{d-1} \cdots U_1 \rho_0 U_1^\dagger \cdots U_{d-1}^\dagger, \\
    H' &= U_{d+1}^\dagger \cdots U_D^\dagger H U_D \cdots U_{d+1},
\end{align}
and $S(\rho', H')$ is a $4 \times 4$ real symmetric matrix whose $(i, j)$-component is defined as
\begin{equation} \label{eq:S_components}
    (S)_{i j}:=\frac{1}{2}{\rm tr}\left[\rho'\left(\varsigma^{\dagger}_{i}H'\varsigma_{j}+\varsigma^{\dagger}_{j}H'\varsigma_{i}\right)\right],
\end{equation}
and can be constructed from expectation values of $H$ calculated using 10 parameter sets, which we call \textit{parameter configuration}~\cite{endo2023optimal}.
The minimizer of $\braket{H}(\bs{q}_d)$ is the eigenvector corresponding to the minimum eigenvalue of the following eigenvalue problem of the $4 \times 4$ matrix $S(\rho', H')$:
\begin{equation}
S(\rho', H')\bs{q}_d = \lambda \bs{q}_d,
\end{equation}
where $\lambda$ is an eigenvalue.
Thus, FQS gives the exact minimizer of the objective function $\braket{H}$ with respect to $\bs{q}_d$, by solving the above eigenvalue problem.

Because the FQS formulation gives a unified form of sequential quantum optimizer 
of PQCs~\cite{wada2022full}, other sequential quantum optimizers can also be reduced to 
an eigenvalue problem.
In NFT~\cite{nakanishi2020} (also in Rotosolve~\cite{ostaszewski2021}), the angle around 
a fixed axis $\bs{n}$ of a single-qubit gate serves as the parameter. 
That is, the parameterized single-qubit gate $U_d$ is represented as
\begin{align}
    U_d = \bs{q}_d \cdot \vec{\varsigma} = \cos \left(\dfrac{\theta_d}{2} \right) I - \iota \sin \left(\dfrac{\theta_d}{2} \right) \bs{n} \cdot \vec{\sigma},
\end{align}
where $\vec{\sigma} = (X, Y, Z)$, and the quaternion $\bs{q}_d$ is now restricted to $(\cos(\theta_d /2), \sin(\theta_d /2)\bs{n}^\top)^\top$.
Then, the expectation $\braket{H}$ can be written as the quadratic form of the vector $\bs{c}_d:=(\cos(\theta_d /2), \sin(\theta_d /2))^\top$, as follows~\cite{wada2022full}:
\begin{align} \label{eq:NFT}
    \braket{H}(\bs{c}_d) = \bs{c}_d^\top \begin{bmatrix}
    S_{00} & \vec{S}_0 \cdot \bs{n} \\
    \vec{S}_0 \cdot \bs{n} & \bs{n}^\top \tilde{S} \bs{n}
    \end{bmatrix}
    \bs{c}_d,
\end{align}
where $\vec{S}:=(S_{01}, S_{02}, S_{03})$ and $\tilde{S}$ is the $3 \times 3$ matrix consisting 
of the lower right components of $S$. 
NFT minimizes the objective function $\braket{H}$ with respect to the angle of 
a single-qubit gate, through the eigenvalue problem of the matrix in Eq.~\eqref{eq:NFT}.
In Rotoselect~\cite{ostaszewski2021}, a finite discrete set of axes is prepared and the angle 
of the single-qubit gate is tried to be updated with respect to each axis in the set by using 
NFT.
Then, the angle and axis that give the minimum objective function is selected.
Thus, Rotoselect adjusts the axis of the single-qubit gates in a discrete way. 
In Fraxis~\cite{watanabe2021}, the axis of a single-qubit gate is to be optimized under the 
condition that its angle is fixed to a constant (typically, $\pi$). 
That is, the parameterized single-qubit gate $U_d$ is represented as 
\begin{align}
    U_d = \bs{q}_d \cdot \vec{\varsigma} = -\iota \bs{n}_d \cdot \vec{\sigma}.
\end{align}
The quaterinion $\bs{q}_d$ is now restricted to $(0, \bs{n}_d)^\top$.
Then, the expectation $\braket{H}$ can be written as the quadratic form of the vector $\bs{n}_d$, as follows~\cite{watanabe2021}:
\begin{align}
    \braket{H}(\bs{n}_d) = \bs{n}_d \tilde{S} \bs{n}_d.
\end{align}
Fraxis minimizes the objective function of the form $\braket{H}$ with respect 
to the axis of a single-qubit gate through the eigenvalue problem of the $3 \times 3$ 
matrix $\tilde{S}$.

To solve the eigenvalue problem, the FQS formulation requires solving a quartic equation to obtain the (local) optimal gate, 
while the Fraxis and NFT require, respectively, solving a cubic and a quadratic equation.
All of these equations for FQS, Fraxis and NFT can be solved analytically~\cite{cardano1560}.

\section{Method} \label{sec:method}

\subsection{Extension of FQS to fractional objective function}
\label{sec:method A}

The problem of finding the minimum eigenvalue of Eq.~\eqref{eq:g_eig} is reduced 
to minimization of the generalized Rayleigh quotient \eqref{eq:g_rayleigh}. 
Here we take the approach using a quantum computer to solve this problem; 
then $\bs{w} \in \mathbb{C}^N$ in Eq.~\eqref{eq:g_rayleigh} is replaced with a quantum 
state vector $\Ket{\psi}$ of $n$-qubit system, and the generalized Rayleigh quotient 
is expressed as 
\begin{align}
    \mathcal{F}(\rho) := \dfrac{\tr{A\rho}}{\tr{B\rho}}, 
\label{eq:objective}
\end{align}
where $\rho=\Ket{\psi (\{\bs{q}_d\}_{i=d}^D)} \Bra{\psi (\{\bs{q}_d\}_{d=1}^D)}$. 
Note that the required number of qubits is $\mathcal{O}(\log_2 N)$, which is thus 
the advantage of using quantum computation. 
If $\log_2 N$ is not an integer, the matrices $A$ and $B$ can be modified so that their dimensions become $2^n \times 2^n$ where $n=\lceil \log_2 N \rceil$ as discussed in 
Appendix~\ref{sec:padding}.

In the formulation of sequential quantum optimization to repeatedly optimize 
$\bs{q}_d$, the expectations can be expressed as Eq.~\eqref{eq:quadratic_form}, 
and thus Eq.~\eqref{eq:objective} is rewritten as 
\begin{align}
     \mathcal{F}(\bs{q}_d) 
        = \dfrac{\bs{q}_d^\top S(\rho', A') \bs{q}_d}{\bs{q}_d^\top S(\rho', B') \bs{q}_d}, 
\label{eq:obj_q}
\end{align}
where $\bs{q}_d$ is the single parameter vector \eqref{single-para-vector}. 
Also $S(\rho', A')$ and $S(\rho', B')$ are the matrices whose components are 
calculated by Eq.~\eqref{eq:S_components} for $A$ and $B$, respectively. 
Because Eq.~\eqref{eq:obj_q} takes the same form as the Rayleigh quotient 
\eqref{eq:g_rayleigh}, $\mathrm{argmin}_{\bs{q}_d} \mathcal{F}(\bs{q}_d)$ 
is identical to the eigenvector corresponding to the minimum eigenvalue of the following 
4-dimensional GEP:
\begin{align}
      S(\rho', A') \bs{p}_i = \lambda_i S(\rho', B') \bs{p}_i, 
\label{eq:eigenprob}
\end{align}
where $\lambda_i$ is the $i$-th eigenvalue and $\bs{p}_i$ is the $i$-th unit eigenvector.
Assuming that the eigenvalues are indexed in the accending order,
$\min_{\bs{q}_d} \mathcal{F}(\bs{q}_d) = \lambda_1$ and 
$\mathrm{argmin}_{\bs{q}_d} \mathcal{F}(\bs{q}_d) 
= \bs{p}_1$.

The entire procedure of the proposed method is shown in Algorithm~\ref{algo:method}. 
The parameters of the single-qubit gates in a PQC are sequentially updated by 
solving the GEP \eqref{eq:eigenprob} until the value of the objective function 
becomes less than a given tolerance value $\varepsilon_\mathrm{tol}$. 
In this study, the order of optimizing the single-qubit gates in line 4 in Algorithm~1 is 
simply chosen as the ascending order, i.e., from top left to bottom right in the circuit diagram of PQC.
We call the procedure from line 3 to 9 in Algorithm~1 an \textit{iteration}, which updates all single-qubit gates once.

Since the matrices $S(\rho', A')$ and $S(\rho', B')$ are constructed by expectation values of $A$ and $B$, they will include the sampling errors, i.e. shot noises.
Due to sampling errors, the matrix $S(\rho', B')$ can be no longer positive definite when the number of sampling is relatively small.
This will cause the numerical instability in solving Eq.~\eqref{eq:eigenprob}.
Thus, if the minimum eigenvalues of $S(\rho', B')$, $\beta_\mathrm{min}$, is negative, we add $(\epsilon - \beta_\mathrm{min}) I$ to $S(\rho', B')$ where $\epsilon$ is a small positive constant. 
This ensures that Eq.~\eqref{eq:eigenprob} is well-posed although the update direction of parameters is slightly changed.

\begin{figure}[t]
\begin{algorithm}[H]
\caption{Sequential optimizer for fractional objective function} \label{algo:method}
\begin{algorithmic}[1]
\Require PQC structure, Hermitians $A$ and $B$, tolerance $\varepsilon_\mathrm{tol}$.
\Ensure Optimized parameters $\{\bs{q}_d^\ast \}_{d=1}^D$.
\State Set initial parameters $\{\bs{q}_d \}_{d=1}^D$ randomly.
\State Set $\mathcal{F}_\mathrm{curr} \leftarrow 0$, $\mathcal{F}_\mathrm{prev} \leftarrow 0$, $\varepsilon \leftarrow 1$
\While {$\varepsilon > \varepsilon_\mathrm{tol}$}
    \For {$d$ in $[1, D]$}
        \State Construct the matrix $S(\rho', A')$ and $S(\rho', B')$.
        \State Update $\bs{q}_d$ by solving Eq.~\eqref{eq:eigenprob}.
        \State Set $\mathcal{F}_\mathrm{prev} \leftarrow \mathcal{F}_\mathrm{curr}$.
        \State Set $\mathcal{F}_\mathrm{curr} \leftarrow \lambda$

        \Comment{$\lambda$ is the minimum or maximum eigenvalue.}

        \State Set $\varepsilon \leftarrow \| \mathcal{F}_\mathrm{curr} - \mathcal{F}_\mathrm{prev} \| / \| \mathcal{F}_\mathrm{prev} \|$
    \EndFor
\EndWhile
\State \Return{$\{\bs{q}_d \}_{d=1}^D$}
\end{algorithmic}
\end{algorithm}
\end{figure}

Let us assume that, through this sequential optimization, we find the set of 
parameters $\{\bs{q}^*_d\}_{d=1}^D$ that exactly minimizes $\mathcal{F}(\rho)$; 
this gives us the solution of GEP in the form of quantum state as 
$\Ket{\bs{v}}=\Ket{\psi (\{\bs{q}^*_d\}_{i=d}^D)}$. 
Note that $\mathcal{O}(N)$ measurements are required to retrieve all the components 
from the quantum state $\Ket{\bs{v}}$. 
Hence, as discussed in \cite{harrow2009quantum}, the proposed method should be 
used in the case where only some characteristic quantities about the solution are 
of interest; typically, such quantity is represented by $\Bra{\bs{v}}M\Ket{\bs{v}}$ 
with $M$ a Hermitian matrix, which can thus be efficiently computed on a 
quantum computer. 
Actually, in Sec.~\ref{sec:results}, we provide two examples where this 
assumption makes sense from an engineering point of view.

Lastly note that, because the FQS formulation gives a unified form of sequential optimizer of 
PQCs~\cite{wada2022full}, other sequential approaches, such as NFT~\cite{nakanishi2020}, 
Rotosolve/Rotoselect~\cite{ostaszewski2021}, and Fraxis~\cite{watanabe2021}, can also be applied 
to solve the GEP problem in the similar approach. 
We indeed use them to compare with FQS in the experiments.  

\subsection{Asymptotic behavior of a parameter update under sampling errors} \label{sec:asymp}
We solve Eq.~\eqref{eq:eigenprob} to update parameters of a single-qubit gate.
Since the matrix $S(\rho', A')$ and $S(\rho', B')$ include sampling errors under a finite number of shots, the resulting eigenvalues and eigenvectors will also include fluctuation.
Here, we summarize the asymptotic behavior of eigenvalues of Eq.~\eqref{eq:eigenprob}.
We provide the detailed analysis in Appendix~\ref{sec:sampling_error}.
In the following, we consider the minimization of the objective function, i.e. the minimum eigenvalue of Eq.~\eqref{eq:eigenprob}.

Let $n_\mathrm{s}$ be the number of shots per individual quantum circuit.
Since $S(\rho', A')$ and $S(\rho', B')$ are respectively constructed by the linear combination of expectation values of $A$ and $B$ calculated by several parameter sets, their perturbations can be represented as
\begin{align}
    S(\rho', A') &= S(\rho', A')^{(0)} + \epsilon S(\rho', A')^{(1)}, \\
    S(\rho', B') &= S(\rho', B')^{(0)} + \epsilon S(\rho', B')^{(1)}, 
\end{align}
where $\epsilon$ is $\mathcal{O}(1 / \sqrt{n_\mathrm{s}})$, the superscript $\cdot^{(0)}$ represents a quantity without any perturbation and $\cdot^{(1)}$ represents one with perturbation.
By considering the second-order asymptotic expansion of eigenvalues $\lambda_i$ and eigenvectors $\bs{p}_i$, we obtain
\begin{align}
    \mathbb{E}\left[ \lambda_i \right] &= \lambda_i^{(0)} + \epsilon^2 \lambda_i^{(0)} \mathbb{E} \left[ \left( \bs{p}_i^{{(0)}\top}  S(\rho', B')^{(1)} \bs{p}_i^{(0)} \right)^2 \right] \nonumber \\
    & - \epsilon^2 \mathbb{E}\left[ \bs{p}_i^{{(1)}\top} \left( S(\rho', A')^{(0)} - \lambda_i^{(0)}S(\rho', B')^{(0)} \right) \bs{p}_i^{(1)} \right].
\end{align}
Therefore, the estimation of the objective function value $\mathcal{F}$ by the minimum eigenvalue $\lambda_i$ has the bias that vanishes asymptotically no slower than or equal to $\epsilon^2$, i.e. $\mathcal{O}(1/n_\mathrm{s})$.
Similarly, we can estimate the objective function value after update of parameters using the perturbed eigenvector, as follows:
\begin{align} \label{eq:perturbed_cost}
    &\mathbb{E}\left[ \mathcal{F}(\bs{p}_1) \right] \nonumber \\
    &\approx \lambda_1^{(0)} + \epsilon^2 \mathbb{E} \left[ \bs{p}_1^{(1)\top} \left( S(\rho', A')^{(0)} - \lambda_1^{(0)} S(\rho', B')^{(0)}\right) \bs{p}_1^{(1)} \right] \nonumber \\
    & \leq \lambda_1^{(0)} + \epsilon^2 \left( \lambda_\mathrm{max}^{(0)} - \lambda_1^{(0)} \right) \mathbb{E} \left[ \bs{p}_1^{(1)\top} S(\rho', B')^{(0)} \bs{p}_1^{(1)} \right],
\end{align}
where $\lambda_\mathrm{max}$ is the maximum eigenvalue.
The second term is related to  the magnitude of the imperfect parameter update due to sampling errors and vanishes asymptotically no slower than or equal to $\epsilon^2$.
This equation indicates that the magnitude of the imperfection depends on the difference between the maximum and minimum eigenvalues of Eq.~\eqref{eq:eigenprob}, i.e. the maximum and minimum objective function values reachable by changing parameters of the single-qubit gate to be updated.
Therefore, when the difference is large, the objective function after parameter update can become large, and vise versa.

\subsection{Generalized eigenvalue problem for a system of linear equations}
\label{sec:method B}

Here we consider a system of linear equations (SLE): 
\begin{align}
K \bs{u} = \bs{f}, 
\label{eq:ls}
\end{align}
where $K$ $\in \mathbb{C}^{N \times N}$ is a given positive-definite matrix, 
$\bs{u} \in \mathbb{C}^N$ is an unknown vector, and $\bs{f} \in \mathbb{C}^N$ is 
a given vector. 
Without loss of generality, we assume that $\| \bs{f} \| = 1$.
Such SLE arises in a variety of applications including partial differential equation~\cite{evans2010partial} and machine learning~\cite{deisenroth2020mathematics,aggarwal2020linear}.

The problem of solving the SLE can be formulated as a GEP as follows:
\begin{align} 
\label{eq:ls_eigenprob}
     \bs{f} \bs{f}^\dagger \bs{v} = \lambda K \bs{v}, 
\end{align}
where $\bs{v} \in \mathbb{C}^N$ is an eigenvector and $\lambda \in \mathbb{R}$ is 
the corresponding eigenvalue. 
Here, $\bs{f}\bs{f}^\dagger$ corresponds to $A$ in Eq.~\eqref{eq:objective}.
Since $\bs{f} \bs{f}^\dagger$ is a rank-1 matrix, this GEP has only one 
non-zero and non-degenerate eigenvalue; the other $N-1$ eigenvalues are all zeros.
Using the non-zero eigenvalue $\hat{\lambda}$ and the corresponding eigenvector 
$\hat{\bs{v}}$, the GEP reads
\begin{align}\label{eq:b_equal}
K \left( \dfrac{\hat{\lambda} \hat{\bs{v}}}{\bs{f}^\dagger \hat{\bs{v}}} \right) 
= \bs{f}.
\end{align}
Since $\hat{\lambda}$ is the non-zero eigenvalue, it is ensured that $\bs{f}^\dagger \hat{\bs{v}} \neq 0$ from Eq.~\eqref{eq:ls_eigenprob}.
Substituting Eq.~\eqref{eq:b_equal} into Eq.~\eqref{eq:ls}, we obtain
\begin{align}
 K\left(\bs{u}-\dfrac{\hat{\lambda} \hat{\bs{v}}}{\bs{f}^\dagger \hat{\bs{v}}}\right) = 0.
\end{align}
Since $K$ is positive-definite, i.e., invertible, we obtain
\begin{align}
\bs{u} = \dfrac{\hat{\lambda} \hat{\bs{v}}}{ \bs{f}^\dagger \hat{\bs{v}} },
\end{align}
meaning that the solution of the SLE is given by the non-zero (maximal) eigenvalue and its 
corresponding eigenvector of the GEP \eqref{eq:ls_eigenprob}.

Therefore, we can employ a quantum computer to solve the SLE \eqref{eq:ls}, by 
formulating it as the GEP \eqref{eq:ls_eigenprob} and using the method described 
in Sec.~\ref{sec:method A}. 
That is, we represent $\bs{f}$ by a quantum state vector 
$\Ket{\bs{f}}$ of $\lceil \log_2 N \rceil$-qubit system, which leads to a GEP 
\eqref{eq:g_eig} with $A=\Ket{\bs{f}} \Bra{\bs{f}}$ and $B=K$. 
Note that the expectation $\tr{A\rho}=\tr{\Ket{\bs{f}} \Bra{\bs{f}} \rho}$ can 
be evaluated as the fidelity of $\rho$ and $\Ket{\bs{f}}$. 
Let us assume that the algorithm described in Sec.~\ref{sec:method A} yields 
the optimal $\{{\bs{q}_d} \}$ and accordingly the optimal 
$\Ket{\psi (\{\bs{q}_d\})}=\Ket{\hat{\bs{v}}}$, which is the quantum-state 
representation of the optimal $\hat{\bs{v}}$. 
This gives us the optimal $\Ket{\bs{u}}=\bs{u}$ as well, if we are just 
interested in the solution up to the constant. 
Otherwise, to have the exact solution $\bs{u}$, we additionally need to calculate 
the value of $\bs{f}^\dagger \hat{\bs{v}}=\braket{\bs{f}|\psi(\{\bs{q}_d\})}$.

Note that, in contrast to the variational quantum algorithms for solving an SLE 
\cite{bravo2019variational, liu2021variational}, the proposed method does not 
require any auxiliary qubit {\it during} the optimizing process of the PQC. 
A brief explanation is as follows. 
In Ref.~\cite{bravo2019variational}, the auxiliary qubit is required to perform 
the Hadamard test and Hadamard-overlap test to update the parameters; also 
Ref.~\cite{liu2021variational} needs to prepare the entangled state 
$(\Ket{0}\Ket{\bs{f}} + \Ket{1}\Ket{\psi}) / \sqrt{2}$, with $\Ket{\psi}$ the 
state generated by a PQC, in order to evaluate the inner product 
$\braket{\bs{f}|\psi}$ and accordingly the cost for updating the parameters. 
On the other hand, as described above, the proposed method generates 
$\hat{\bs{v}}=\Ket{\hat{\bs{v}}}$ without any auxiliary qubit. 
If one needs $\bs{u}=\Ket{\bs{u}}$, the proposed method also requires an 
auxiliary qubit to calculate the inner product 
$\bs{f}^\dagger \hat{\bs{v}}=\braket{\bs{f}|\psi}$ on the quantum device, but 
this operation is necessary only once after the entire optimization process.

\subsection{Complexity and Resource}

Let us assume that $A$ and $B$ are band-matrices with the bandwidth of $k_A$ and $k_B$, 
respectively, which typically appear in the problem of finite element method (FEM). 
The proposed method calculates the expectation values of $A$ and $B$, 
which requires $\mathcal{O}(n k)$ kinds of quantum circuits using the 
extended Bell measurement (XBM) technique~\cite{kondo2022computationally}, where 
$n = \lceil \log_2 N \rceil$ and $k:=\mathrm{max}(k_A, k_B)$.
An overview of XBM is given in Appendix~\ref{sec:xbm}.
Suppose the number of shots per quantum circuit is $s$.
Then, the total number of shots required to calculate expectation values of $A$ and $B$ is $\mathcal{O}(n k s)$.

Also, as mentioned below Eq.~\eqref{eq:objective}, the proposed method has a quantum 
advantage that it uses only $\lceil \log_2 N \rceil$-qubits to represent a vector 
in $\mathbb{C}^N$. 
Thus, even for practical problems using the FEM with tens or hundreds of thousands 
of degrees of freedom~\cite{ribeiro2013finite, muhammad2020finite, belhocine2020thermomechanical}, it requires only less than 20 qubits. 

To encode the right hand side vector $\bs{f}$ into a quantum state $\Ket{\bs{f}}$, we have to design the so-called \textit{oracle} $U_f$ that prepares $\Ket{\bs{f}} = U_f \Ket{0}^{\otimes n}$.
When $\bs{f}$ corresponds to a relatively simple input representing such as a point source or uniform input, the oracle can be intuitively designed using Pauli-$X$ and Hadamard gates, as we describe in Sec.~\ref{sec:poisson_result}.
In general cases, on the other hand, some amplitude encoding techniques~\cite{Zhang2021low, nakaji2022approximate} are required.

\section{Numerical Experiments} 
\label{sec:results}

In the following, we provide numerical experiments.
Unless otherwise stated, we used the statevector simulator of Qiskit~\cite{Qiskit}.

\subsection{Solving the Poisson equation}

\subsubsection{Problem statement}

We apply the proposed method to the problem of solving a partial differential equation 
(PDE). 
Among PDEs, we here focus on the Poisson equation, which appears in versatile applications 
including steady-state heat transfer, electrostatics~\cite{griffiths1999introduction}, 
and computational fluid dynamics~\cite{chung2010computational, blazek2015computational}. 
Before proceeding, recall that the proposed method obtains the solution vector as 
a quantum state $\Ket{\psi \{\bs{q}_d \}}$, meaning that practically we can retrieve 
only a few characteristic quantities from it. 
For the case of PDE problem, such partial information is for instance the surface 
temperature of a material, which indeed can be calculated from  a few component of 
the entire solution vector of the Poisson equation.

Let $\Omega \subset \mathbb{R}^m$ denote an open bounded set where $m$ is the number 
of spatial dimensions. 
The Poisson equation governs the state field $u(\bs{x}) \in \mathbb{C}$ at the 
spatial coordinate $\bs{x} \in \Omega$, which behaves as 
\begin{align} \label{eq:poisson}
	- \nabla^2 u(\bs{x}) = f(\bs{x}) \quad \text{for } \bs{x} \in \Omega,
\end{align}
where $\nabla$ is the gradient operator with respect to $\bs{x}$ and $f(\bs{x}): \Omega \rightarrow \mathbb{C}$ is a given function.
We impose the Dirichlet boundary condition on $\partial \Omega$ as
\begin{align}
    u(\bs{x}) = 0 \quad \text{on } \bs{x} \in \partial \Omega.
\end{align}
Discretizing the Poisson equation by FEM~\cite{allaire2007numerical, hughes2012finite} 
yields an SLE written as follows:
\begin{align}
    K \bs{U} = \bs{F}, 
    \label{eq:ls_poisson}
\end{align}
where $K \in \mathbb{R}^{N \times N}$ is a positive-definite matrix called the 
stiffness matrix; 
also, $\bs{U} \in \mathbb{C}^N$ and $\bs{F} \in \mathbb{C}^N$ are the 
discretized vectors of $u(\bs{x})$ and $f(\bs{x})$, respectively. 
$N$ is the number of nodes of the finite element mesh. 
The discretization procedure by FEM is detailed in Appendix~\ref{sec:fem_poisson}.

The SLE \eqref{eq:ls_poisson} has the form of Eq.~\eqref{eq:ls}, and thus, it can be 
formulated as a GEP and solved using a quantum computer. 
The procedure is summarized in Algorithm~\ref{algo:method}; in our case, 
$A=\bs{F}\bs{F}^\dagger$ and $B=K$. 
In particular, because $K$ is a band-matrix, the expectation value 
$\tr{B \rho}=\tr{K \rho}$ can be efficiently calculated by 
XBM~\cite{kondo2022computationally}. 
Also, due to the linearity of Eq.~\eqref{eq:ls_poisson}, 
we can set $\|\bs{F} \|=1$ to well define the quantum state $\Ket{\bs{F}}$. 
Then, assuming that $\Ket{\bs{F}}$ is efficiently prepared by a unitary $U_F$, i.e. 
$\Ket{\bs{F}} = U_F \Ket{0}^{\otimes n}$, we can use the 
inversion test~\cite{ruan2021quantum} to calculate 
$\mathrm{tr}(A\rho)=\mathrm{tr}(\Ket{\bs{F}}\Bra{\bs{F}} \rho)$.

Here, we focus on the one-dimensional Poisson equation discretized by the first-order 
elements whose length are uniformly $1$. 
We use $32=2^5$ nodes for discretization, which requires 5-qubits. 
As a test case, we herein set the right hand side of the Poisson equation, $f(\bs{x})$, 
to a step function given in the form of quantum state as 
\begin{align}
\Ket{\bs{F}} = \dfrac{1}{2^{n/2}} \sum_{j=0}^{N-1} (-1)^{j_{n-1}} \Ket{j},
\end{align}
where $j_{n-1}$ is the value of the most significant bit (MSB) of the unsigned binary 
representation of $j$.
This quantum state $\Ket{\bs{F}}$ can be efficiently prepared by the the following unitary 
$U_F$:
\begin{align}
U_F = H^{\otimes n} \left( X \otimes I^{\otimes (n-1)} \right).
\end{align}
Note that, since $\Ket{\bs{F}}$ is a real vector in $\mathbb{R}^{2^n}$ and $K$ is a real 
matrix, the solution of this problem also lies in the real space.

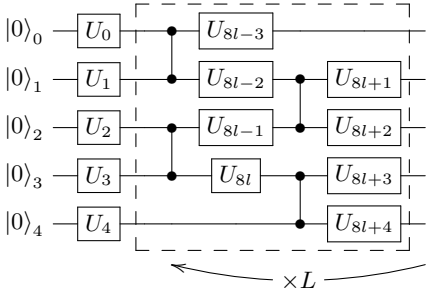
\begin{figure}[t]
    \centering
    \begin{tabular}{c}
    \Qcircuit @C=1.0em @R=.5em {
      \lstickx{\ket{0}_0} & \gate{U_{0}} & \push{\rule{0em}{1.5em}} \qw & \ctrl{1}     & \gate{U_{8l-3}} & \qw          & \qw             & \qw & \\
      \lstickx{\ket{0}_1} & \gate{U_{1}} & \qw & \control \qw & \gate{U_{8l-2}} & \ctrl{1}     & \gate{U_{8l+1}} & \qw & \\
      \lstickx{\ket{0}_2}  & \gate{U_{2}} & \qw & \ctrl{1}     & \gate{U_{8l-1}} & \control \qw & \gate{U_{8l+2}} & \qw & \\
      \lstickx{\ket{0}_3} & \gate{U_{3}} & \qw & \control \qw & \gate{U_{8l}}   & \ctrl{1}     & \gate{U_{8l+3}} & \qw & \\
      \lstickx{\ket{0}_4} & \gate{U_{4}} & \qw & \qw          & \qw             & \control \qw & \gate{U_{8l+4}} & \qw & \\
      & & & & & & &    \arrep{llll}
      \gategroup{1}{3}{5}{7}{.7em}{--}
    }
    \end{tabular}
    \caption{Alternating layered PQC for 5 qubits}
    \label{fig:ALA_PQC}
\end{figure}

\subsubsection{Results and Discussion} \label{sec:poisson_result}
\paragraph{Dependency of results on initial parameters}
\begin{figure}[t]
    \centering
    \includegraphics[width=\columnwidth]{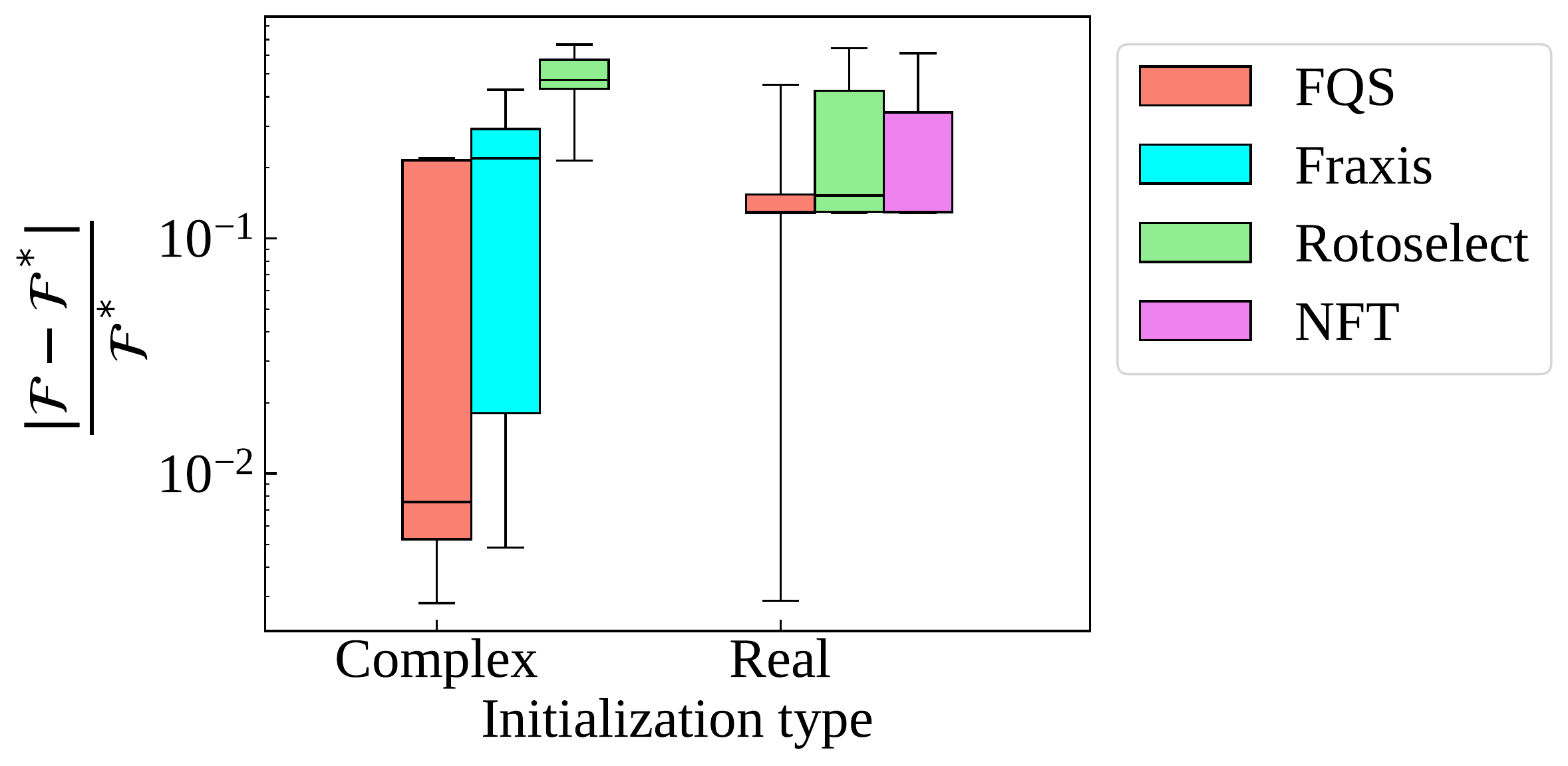}
    \caption{Boxplots of the relative error of the resulting objective function 
    with respect to the exact minimum value ${\mathcal F}^*$. 
    VQA were conducted on the alternating layered PQC of $L=2$. The 30 independent trials were performed for each of the real and the complex space initialization methods.}
    \label{fig:result_5qubits_2layers_ALA}
\end{figure}

\begin{figure*}[t]
    \centering
    \subfloat{\includegraphics[width=0.33\textwidth]{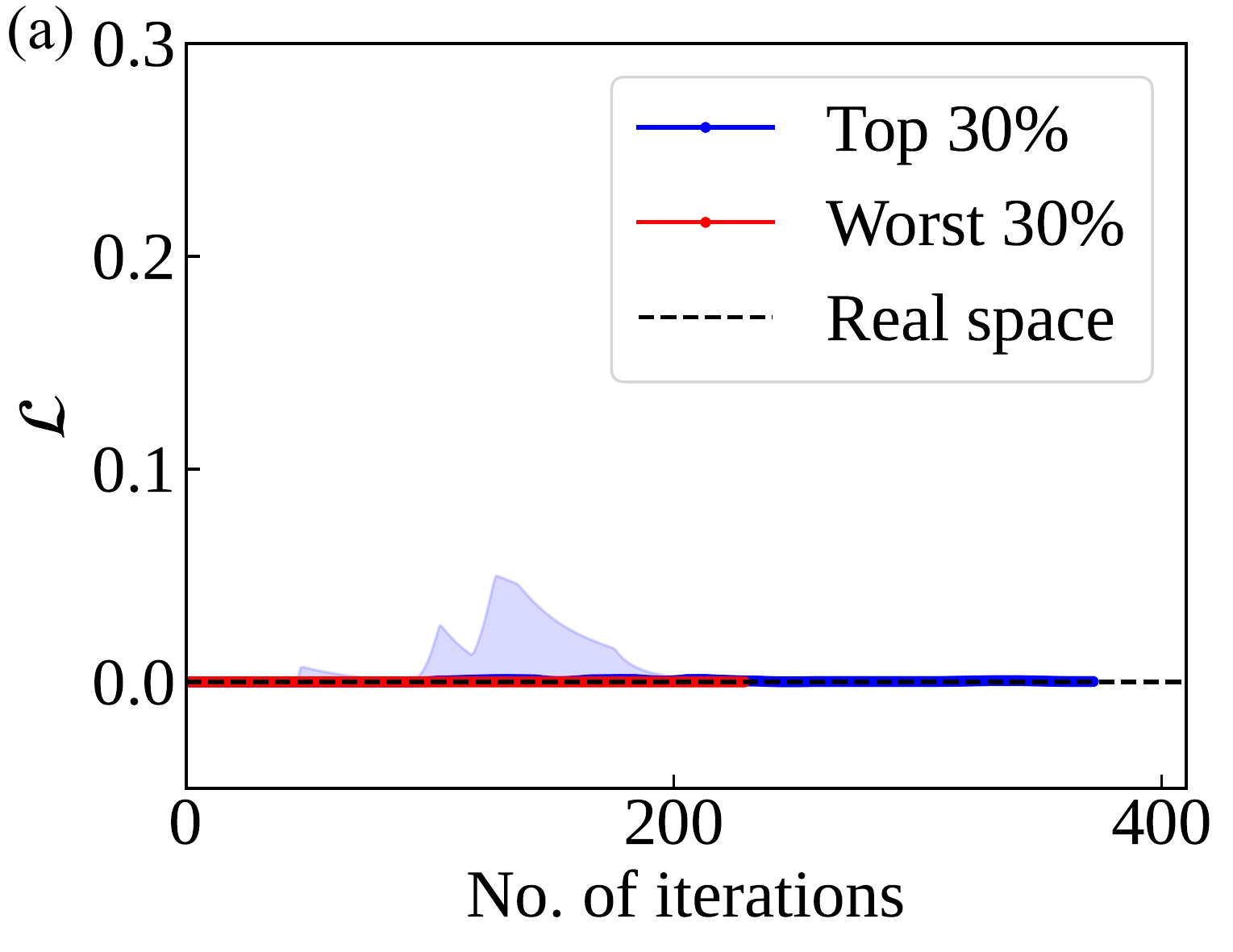}}
    \subfloat{\includegraphics[width=0.33\textwidth]{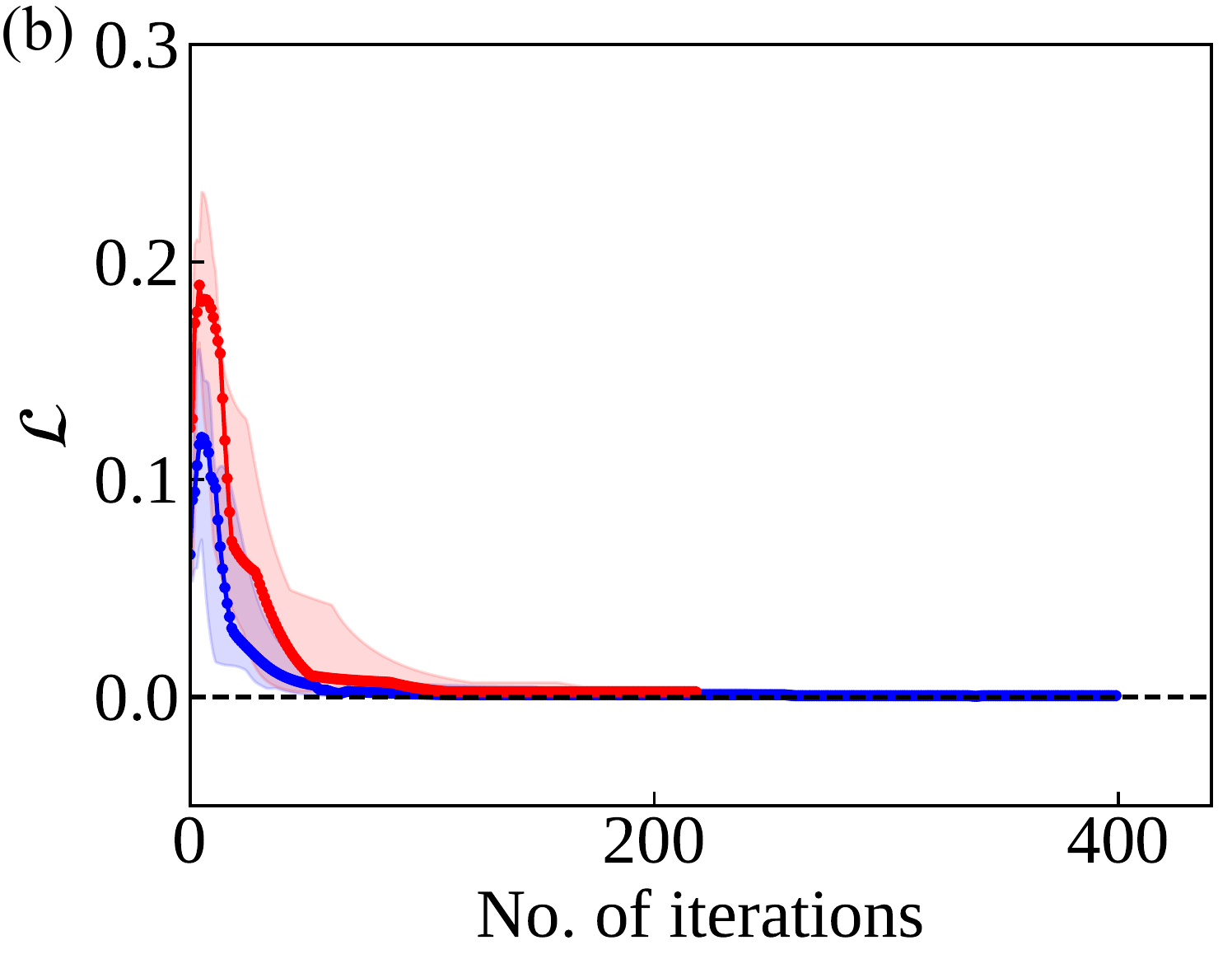}}
    \subfloat{\includegraphics[width=0.33\textwidth]{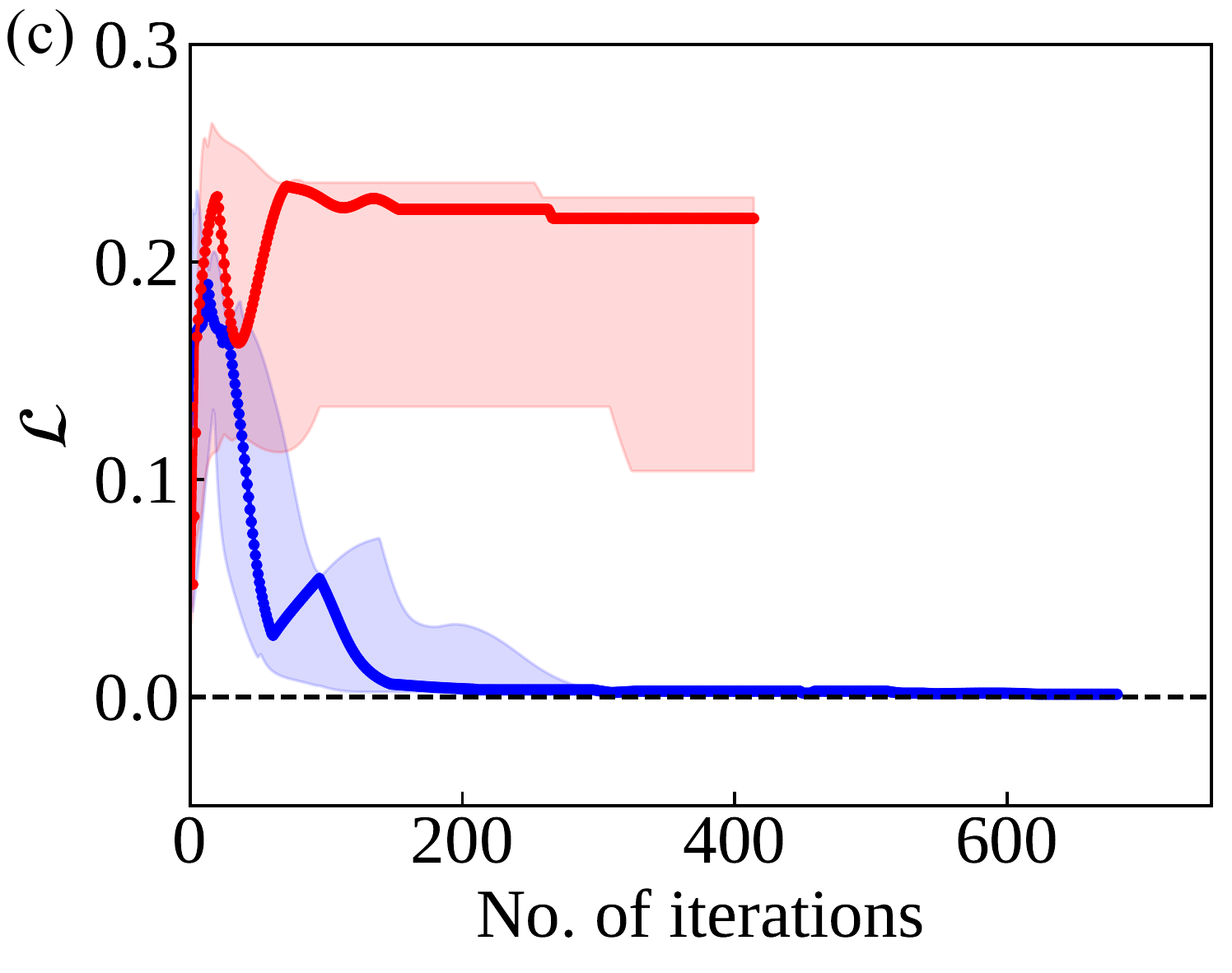}}
    \caption{The distance from the state vector to the real space for the result of (a) FQS with real space initialization, (b) FQS with complex space initialization, and (c) Fraxis.
    The blue plots are the median values of the top 30\% solutions among the 30 independent trials and the blue shaded ranges illustrate the 25- and 75-percentiles. The red plots are the median values of the worst 30\% solutions among the trials and the red shaded regions represent the 25- and 75-percentiles.}
    \label{fig:result_phase_fqs_5qubits_2layers_ALA}
\end{figure*}

This part shows the dependency of results on the initial parameters.

We applied the proposed method with FQS, Fraxis, Rotoselect, and NFT to solving the 
above-described Poisson equation. 
For NFT, we use the $R_y$ gate $R_y(\theta)=e^{-i\theta Y}$ with $Y$ the 
Pauli-$Y$ matrix, for all parameterized single-qubit gates, where the angles of these gates 
are randomly initialized.
For Fraxis, the rotational axes of all single-qubit gates were set randomly on the unit 
spherical surface. 
For FQS and Rotoselect, we employed two initialization strategies termed the 
{\it real space initialization} and the {\it complex space initialization}. 
In the real space initialization, the rotational axes of all single-qubit gates were set to $y$-axis in the beginning of VQA, while their initial angles were randomly generated both for Rotoselect and FQS. 
In the complex space initialization for Rotoselect, the initial rotational axis of single-qubit gates were randomly selected from $x$, $y$, or $z$-axes and the angles of the gates also randomly initialized.
The complex space initialization for FQS randomly set initial unit quaternions; that is, 
a quaternion is sampled uniformly from the three-dimensional unit hyper-sphere. 
We employed the alternating layered PQC~\cite{cerezo2021cost} shown in 
Fig.~\ref{fig:ALA_PQC} and set the number of layers to $L=2$. 
The 30 independent trials with both initializations were performed to examine the 
statistical behaviors of the proposed method.

Figure~\ref{fig:result_5qubits_2layers_ALA} shows boxplots of the relative errors of the 
resulting objective function value after optimization, compared to the exact minimum 
value ${\mathcal F}^*$. 
FQS with the complex space initialization clearly exhibits better capability of obtaining 
lower objective function values, which implies that the use of complex space is effective 
even for problems involving the solution lying in the real space.

To comprehend the behaviors of searching in the complex space more clearly, we analyzed the distance of the state vector during the optimization process from the real space.
Let $r_i$ and $c_i$ denote the real and imaginary parts of the $i$-th complex amplitude of the state vector, respectively. 
When the state vector lies in the real space up to a global phase, all the points $\{(r_i, c_i)\}_{i=0}^{2^n-1}$ on the complex plane lie in a straight line.
Thus, the distance from the state vector to the real space can be evaluated by the distance from the point set $\{(r_i, c_i)\}_{i=0}^{2^n-1}$ to a line.
Let $X$ be a matrix in $\mathbb{R}^{2^n \times 2}$ whose the first and the second columns are $(r_0, \ldots, r_{2^n - 1})^\top$ and $(c_0, \ldots, c_{2^n - 1})^\top$, respectively.
Then, the distance from the point set $\{(r_i, c_i)\}_{i=0}^{2^n-1}$ to a line in the complex plane can be calculated as the minimum eigenvalue of $X^\top X$ based on the principle component analysis (PCA).
Using the two eigenvalues denoted by $\mu_1 > \mu_2$, we define the metric of the distance from the state vector to the real space as
\begin{align}
    \mathcal{L} := \dfrac{\mu_2}{\mu_1 + \mu_2}.
\end{align}
Figures~\ref{fig:result_phase_fqs_5qubits_2layers_ALA} shows the histories of the distance $\mathcal{L}$ of the state vector during the optimization by FQS and Fraxis for the top $30\%$ and the worst $30\%$ trials among 30 trials in terms of the resulting objective function values.
Note that the number of iterations until convergence is not necessarily consistent among 30 trials. 
Hence, we completed the shorter trajectories by the final value at the last iteration such that the lengths of the trajectories were equalized. 
As shown in Fig.~\ref{fig:result_phase_fqs_5qubits_2layers_ALA}(a), the top $30\%$ trials starting from the real space initialization used the complex space in optimization process, while the states of the worst $30\%$ trials stayed in the real space. 
This result implies that the complex space may provide a bypass between the initial and the optimal states even though they are both in the real space. 
On the other hand, Fig.~\ref{fig:result_phase_fqs_5qubits_2layers_ALA}(b) indicates that the worst $30\%$ trials with the complex space initialization had the slightly larger distance to the real space than the top $30\%$ trials, which deteriorated the optimization.
Figure~\ref{fig:result_phase_fqs_5qubits_2layers_ALA}(c) shows that the top $30\%$ trials by Fraxis converged in the real space while the worst $30\%$ trials by Fraxis got stuck with the significantly large distance to the real space.
Compared to FQS and Fraxis, we also observed that the distances to the real space were hardly changed during optimization by Rotoselect. 
Thus, the complex space did not bring benefits to the Rotoselect optimization, which implies 
that the continuous adjustment of the rotational axis, rather than the discrete one, is 
important to make use of the complex space for enhancing the trainability of PQC.

\paragraph{Dependency of results on the number of layers}
\begin{figure}
    \centering
    \begin{tabular}{c}
    \Qcircuit @C=.3em @R=0.5em {
      \lstickx{\ket{0}_0} & \qw & \gate{U_{0}} & \push{\rule{0em}{1.5em}} \qw & \ctrl{1}     & \qw           & \qw          & \qw             & \qw          & \qw             & \qw          & \qw             & \control \qw  & \gate{U_{5l+4}} & \qw             & \qw \\
      \lstickx{\ket{0}_1} & \qw & \gate{U_{1}} & \qw & \control \qw & \gate{U_{5l}} & \ctrl{1}     & \qw             & \qw          & \qw             & \qw          & \qw             & \qw           & \qw             & \gate{U_{5L+5}} & \qw  \\
      \lstickx{\ket{0}_2} & \qw & \gate{U_{2}} & \qw & \qw          & \qw           & \control \qw & \gate{U_{5l+1}} & \ctrl{1}     & \qw             & \qw          & \qw             & \qw           & \qw             & \gate{U_{5L+6}} & \qw  \\
      \lstickx{\ket{0}_3} & \qw & \gate{U_{3}} & \qw & \qw          & \qw           & \qw          & \qw             & \control \qw & \gate{U_{5l+2}} & \ctrl{1} \qw & \qw             & \qw           & \qw             & \gate{U_{5L+7}} & \qw   \\
      \lstickx{\ket{0}_4} & \qw & \gate{U_{4}} & \qw & \qw          & \qw           & \qw          & \qw             & \qw          & \qw             & \control \qw & \gate{U_{5l+3}} & \ctrl{-4} \qw & \qw             & \gate{U_{5L+8}} & \qw  \\
      &&&&&&&&&&&&& \arrep{lllllllll}
      \gategroup{1}{4}{6}{14}{.25em}{--}
    }
    \end{tabular}
    \caption{Cascading-block PQC for 5 qubits}
    \label{fig:Cascading_PQC}
\end{figure}
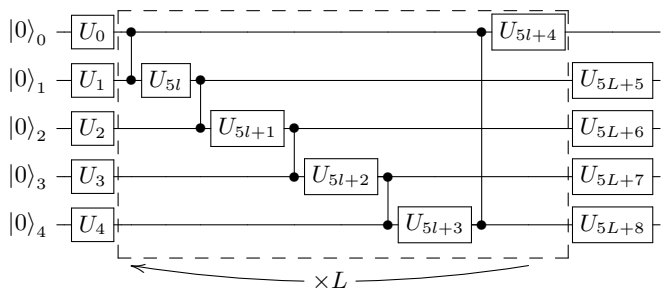
\begin{figure}[t]
    \centering
    \subfloat{\includegraphics[width=\columnwidth]{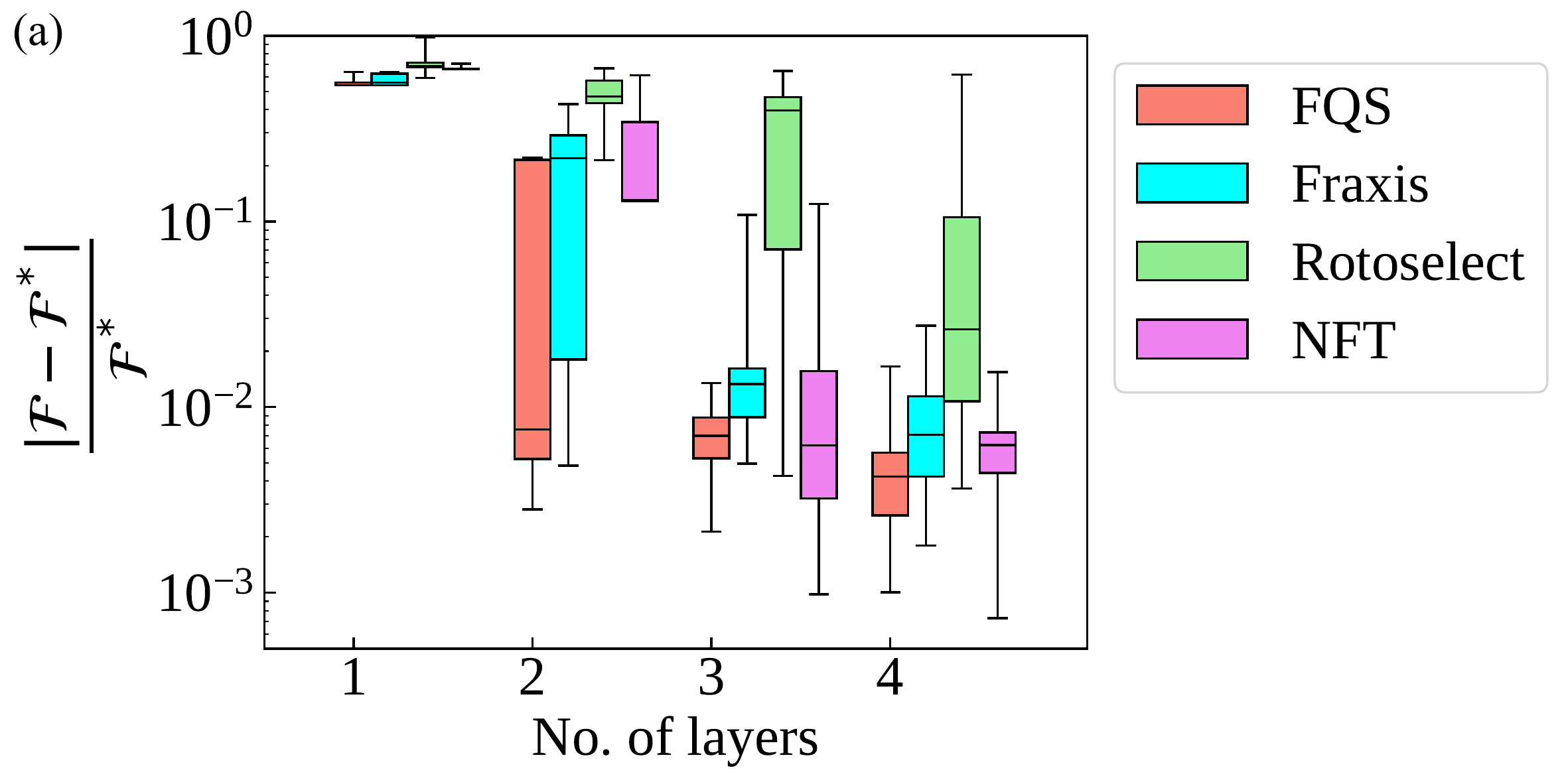}} \\
    \subfloat{\includegraphics[width=\columnwidth]{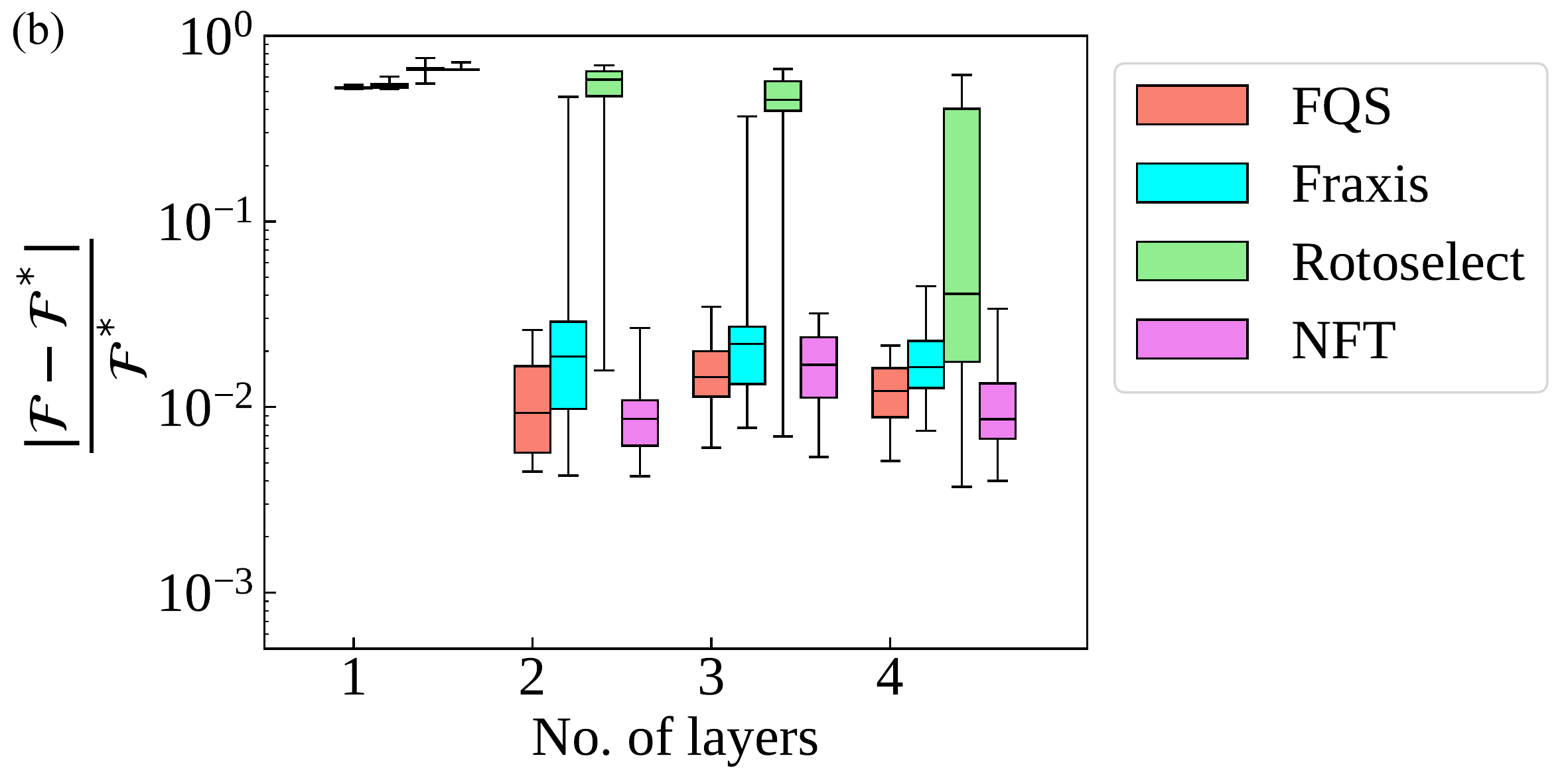}}
    \caption{Boxplots of the relative error between the resulting objective function and 
    the exact minimum value ${\mathcal F}^*$ for (a) alternating layered PQC and (b) Cascading block PQC. 
    The 30 independent trials were performed for each setting of the number of layers.}
    \label{fig:result_5qubits}
\end{figure}
Next, we examined the dependency on the number of layers of PQCs.
Here, we employed two kinds of PQCs; one is the alternating layered PQC in Fig.~\ref{fig:ALA_PQC} and the other is the cascading block PQC in Fig.~\ref{fig:Cascading_PQC}.
The 30 independent trials for NFT, Rotoselect, Fraxis and FQS were performed for each condition.
For Rotoselect and FQS, the complex space initialization was used.

Figure~\ref{fig:result_5qubits} shows boxplots of the relative errors of the objective function with respect to the exact minimum value for alternating layered and cascading block PQCs.
As shown in Fig.~\ref{fig:result_5qubits}(a), the benefits to use the complex search space were conspicuous when the number of layers is relatively small while the converged objective function values obtained by FQS and NFT becomes closer as the number of layers becomes larger.
This is because the expressibility of PQCs become high enough to express the solution even in the real space.
Compared with Fraxis, FQS led to better solutions, which implies that the simultaneous optimization of angles and axes of gates contributes to the efficient shortcut between the initial and the target states using the complex space.
On the other hand, Fig.~\ref{fig:result_5qubits}(b) shows that the objective function values converged by FQS, Fraxis and NFT exhibit no siginificant differences, which means that the effectiveness of the use of the complex space depends on PQCs.
Therefore, we conclude that the continuously controllable angle and axis promotes the effective use of the complex space, which leads to higher quality solutions especially with the shallow depth for a certain PQC structure.
The feature of the PQC that can effectively bring out the benefits to use the complex space should be identified in the future research.

\paragraph{Evaluation by the QASM simulator}
\begin{figure*}[t]
    \centering
    \subfloat{\includegraphics[height=0.15\textheight]{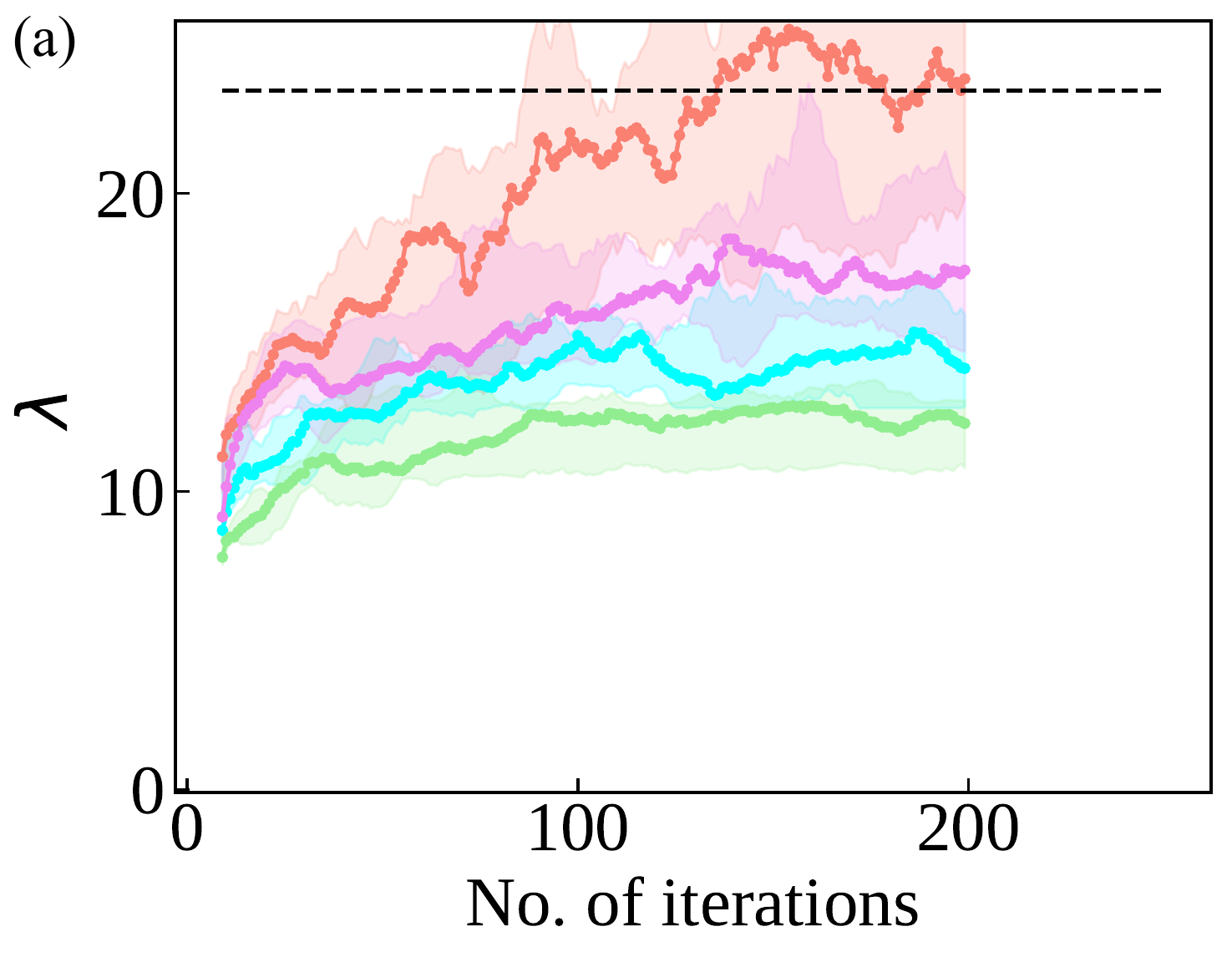}}
    \subfloat{\includegraphics[height=0.15\textheight]{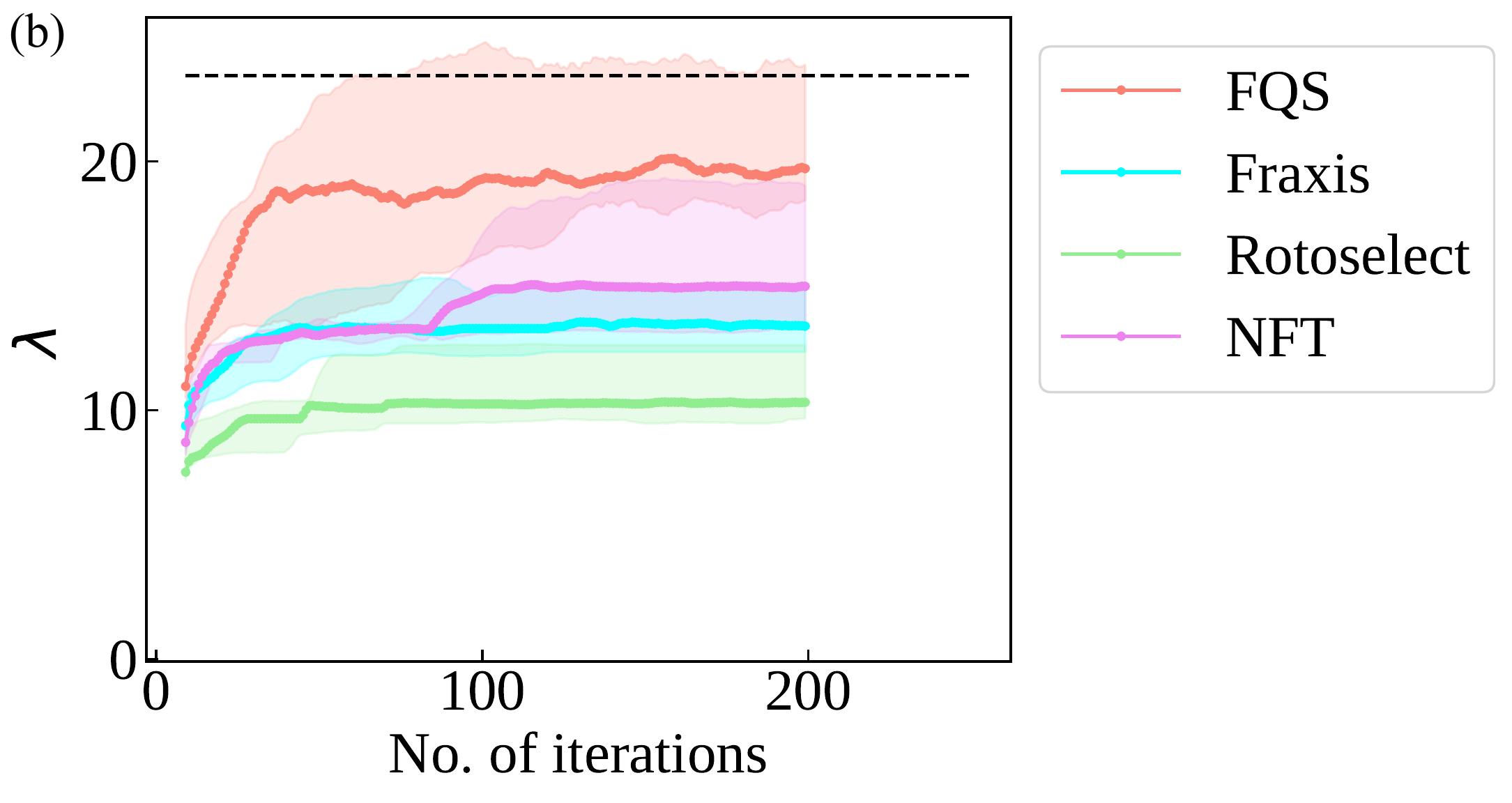}}
    \caption{Histories of the moving average of the maximum eigenvalue of Eq.~\eqref{eq:eigenprob} with the window length of $10$ for (a) $n_\mathrm{s}=10^4$ and (b) $n_\mathrm{s}=10^5$. $n_\mathrm{s}$ is the number of shots for individual quantum circuit. Plots represent the median value and shaded ranges represent ranges from 25- to 75-percentiles of the 10 independent trials. Dashed-line represents the exact optimal value of the objective function.}
    \label{fig:result_5qubits_QASM_lam}
\end{figure*}
\begin{figure*}[t]
    \centering
    \subfloat{\includegraphics[height=0.15\textheight]{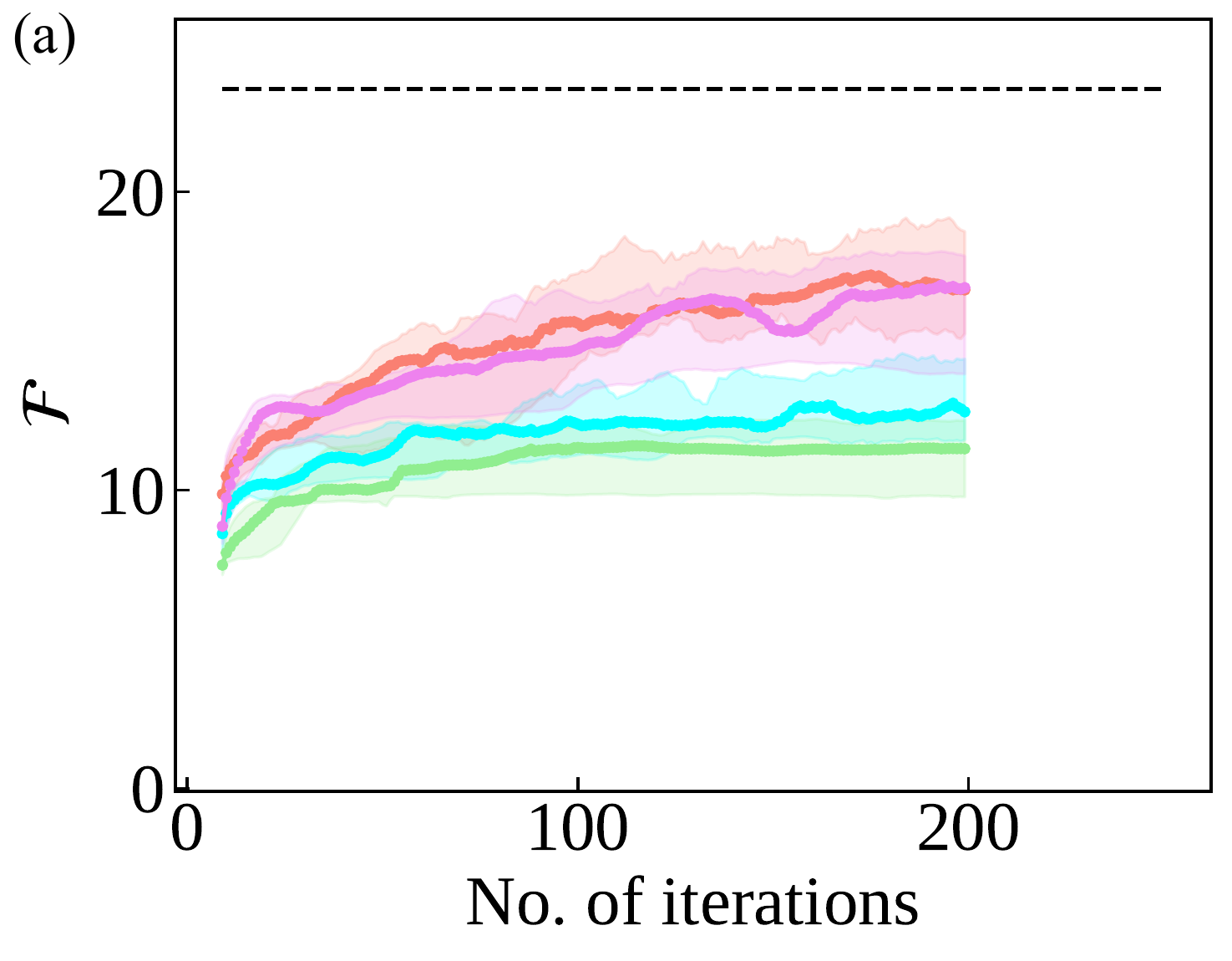}}
    \subfloat{\includegraphics[height=0.15\textheight]{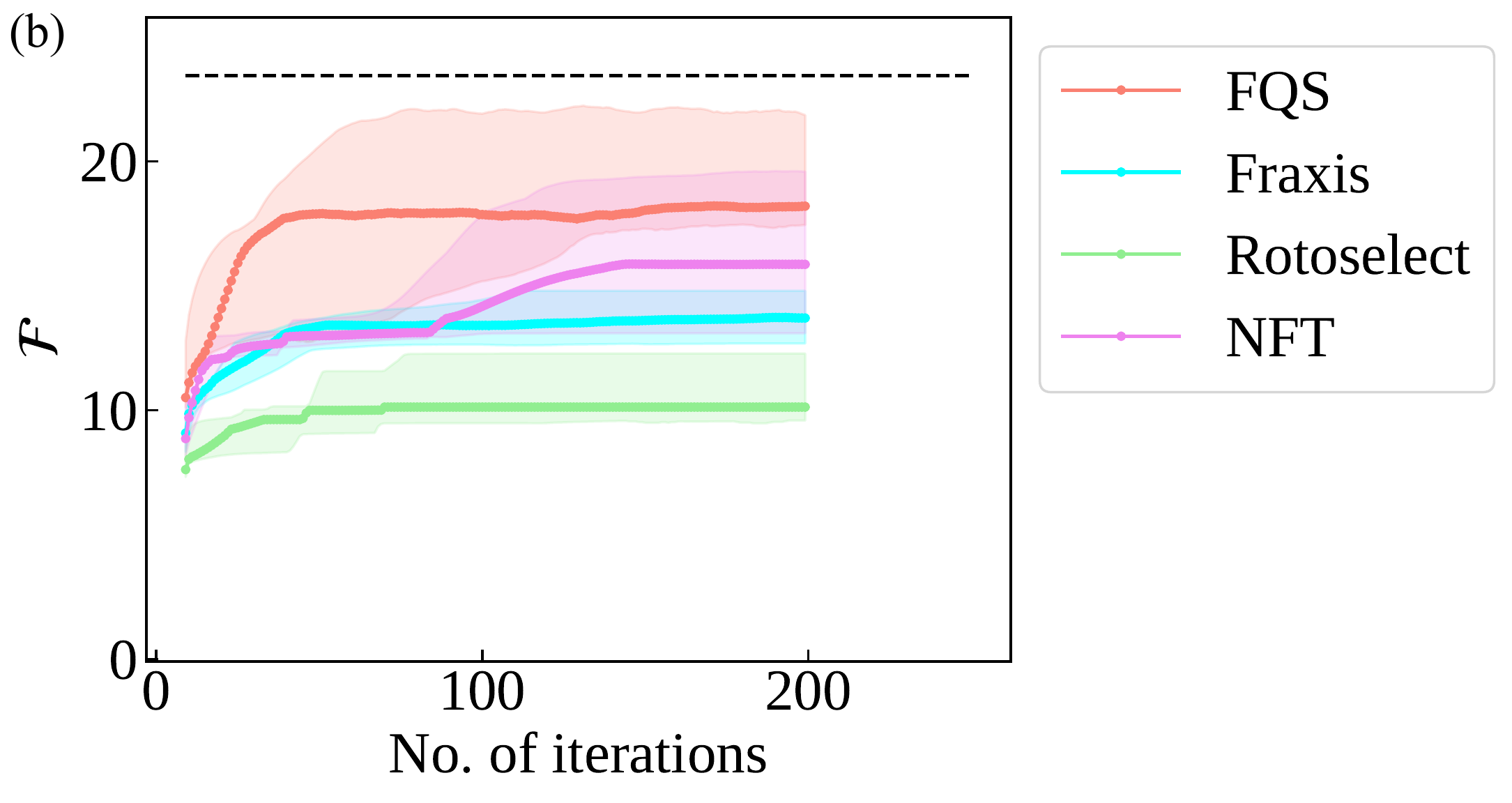}}
    \caption{Histories of the moving average of the objective function value $\mathcal{F}$ with the window length of $10$, which is reevaluated by the statevector simulator for the parameter history of each 10 independent trial in Fig.~\ref{fig:result_5qubits_QASM_lam}~(a) and (b). Plots represent the median value and shaded ranges represent ranges from 25- to 75-percentiles of the 10 independent trials. Dashed-line represents the exact optimal value of the objective function.}
    \label{fig:result_5qubits_QASM_cost}
\end{figure*}
Here, we examined the proposed method under finite-shot settings to study the effect of sampling errors by the QASM simulator in Qiskit~\cite{Qiskit}.
To construct the matrices $S(\rho', A')$ and $S(\rho', B')$ in Eq.~\eqref{eq:eigenprob}, we used the optimal parameter configuration, which can mitigate the sampling errors due to the geometrical symmetry of parameter configurations~\cite{endo2023optimal}.
Although Ref.~\cite{endo2023optimal} obtained the optimal parameter configuration for eigenvalue problems, we confirmed that it is also applicable for generalized eigenvalue problems, as stated in Appendix~\ref{sec:param_config}.
We employed the alternating layered PQC in Fig.~\ref{fig:ALA_PQC} with two layers.
We used the complex space initialization for Rotoselect and FQS.
We set the maximum number of iterations of the proposed method to $200$ and performed the 10 independent trials for NFT, Rotoselect, Fraxis and FQS.

Figure~\ref{fig:result_5qubits_QASM_lam} shows the history of the moving average of the maximum eigenvalue of Eq.~\eqref{eq:eigenprob} with the window length of $10$.
We observe the history became smoother as $n_\mathrm{s}$ increased.
These history indicated that the result of FQS reaches larger eigenvalues than those of the other optimizers, which again demonstrated the benefits to use the complex search space.

We also observe that the maximum eigenvalues exceeded the exact optimal value shown by dashed-line several times.
This would be because the eigenvalue whose matrix components are estimated by finite measurement outcomes includes the bias as discussed in Section~\ref{sec:asymp}.
Figure~\ref{fig:result_5qubits_QASM_cost} shows the history of the objective function value ${\mathcal F}$, which is reevaluated by the statevector simulator for the parameter history of each 10 independent trial in Fig.~\ref{fig:result_5qubits_QASM_lam}.
As $n_\mathrm{s}$ increased, the errors between the estimated objective function value, i.e. the maximum eigenvalue in Fig~\ref{fig:result_5qubits_QASM_lam} and the actual objective function value in Fig.~\ref{fig:result_5qubits_QASM_cost}  became small. 
We also observed that the advantage of FQS over NFT was not distinct for $n_\mathrm{s}=10^4$.
This would be because the effect of sampling errors on the objective function after parameter update in FQS is larger than that in NFT due to the larger difference between the maximum and minimum eigenvalues in FQS compared to NFT in Eq~\eqref{eq:perturbed_cost}, which comes from that fact that the variational space of FQS includes that of NFT.
That is, FQS has more potential to improve the objective function, but also has a susceptibility to sampling errors.
In the present study, we set the same number of shots to measure all quantum circuits for simplicity.
We hope to conduct our future work to allocate the optimal number of shots for each quantum circuit to mitigate the sampling errors.

\subsection{Eigen-frequency problem of a linear elastic solid}

\subsubsection{Problem statement}

Next, let us consider the problem of finding the lowest eigen-frequency of a linear elastic 
solid, which is important in engineering to design mechanical structures for maximally 
improving their dynamic behaviors~\cite{ma1994structural}.
Because this application focuses on finding the lowest eigen-frequency that is obtained as 
the minimum of the objective function by the proposed method, one does not have to retrieve 
all the components of the optimized quantum state $\ket{\psi (\{ \bs{q}_d \})}$.

Let $\Omega \subset \mathbb{R}^m$ denote an open bounded set where $m$ is the number of spatial dimensions.
The eigenvalue problem of a linear elastic solid occupying $\Omega \subset \mathbb{R}^m$ consists of the stress equilibrium equation, the constitutive law, and the displacement-strain relationship, which are respectively given as
\begin{align} \label{eq:solid}
\begin{cases}
    & \nabla \cdot \bs{\sigma(\bs{x})} = -\varrho \omega^2 \bs{u}(\bs{x}), \\
    & \bs{\sigma(\bs{x})} = \bs{C} : \bs{\varepsilon(\bs{x})}, \\
    & \bs{\varepsilon(\bs{x})} = \dfrac{1}{2} \left( \nabla \bs{u} + \nabla \bs{u}^\top \right).
\end{cases}
\end{align}
Here, $\bs{x} \in \Omega$ is the spatial coordinate, $\varrho$ is the density of the solid, $\omega$ is the eigenfrequency, $\bs{u} \in \mathbb{R}^m$ is the displacement field, $\bs{\sigma}$ is the stress tensor, $\bs{C}$ is the elastic tensor, and $\bs{\varepsilon}$ is the strain tensor.
The operator $:$ represents the double dot product, which acts as $\bs{C}:\bs{\varepsilon} = C_{ijkl} \varepsilon_{kl}$ with the Einstein's summation convention.
Assuming that the solid is an isotropic material, the elastic tensor can be represented using the Young's modulus $E$ and the Poisson ratio $\nu$ as
\begin{align} \label{eq:elastic_tensor}
C_{ijkl} = \dfrac{E \nu}{(1 + \nu)(1 - 2\nu)} \delta_{ij} \delta_{kl} + \dfrac{E}{2(1 + \nu)} \left( \delta_{ik} \delta_{jl} + \delta_{il} \delta_{jk} \right),
\end{align}
where $\delta_{ij}$ is the Kronecker's delta.
The Young's modulus represents the modulus of elasticity in tension or compression of 
a solid material. 
The Poisson ratio, which is in the range of $-1 < \nu < 0.5$, represents a measure of the deformation of a material in directions perpendicular to the specific direction of loading.
These equations compose the governing equation of the displacement field $\bs{u}$.
We consider that the Dirichlet boundary condition is imposed on the boundary $\Gamma_\mathrm{D} \subset \partial \Omega$ as
\begin{align}
    \bs{u}(\bs{x}) = \bs{0} \quad \text{on}~ \bs{x} \in \Gamma_\mathrm{D}.
\end{align}
Discretizing this governing equation by FEM~\cite{allaire2007numerical, hughes2012finite} yields a GEP as follows:
\begin{align}
K \bs{U} = \lambda M\bs{U},
\end{align}
where $K \in \mathbb{R}^{(Nm) \times (Nm)}$ is the stiffness matrix, $M \in \mathbb{R}^{(Nm) \times (Nm)}$ is the mass matrix, $\bs{U} \in \mathbb{R}^{Nm}$ is the displacement field vector, and $\lambda=\omega^2$ is the eigenvalue.
The discretization procedure by FEM is detailed in Appendix~\ref{sec:fem_solid}. 
Both $K$ and $M$ are positive-definite. 
Hence, the proposed method can be applied to analyzing the eigen-frequency of a linear 
elastic solid. 
In particular, $\tr{B \rho}=\tr{M \rho}$ and $\tr{A \rho}=\mathrm{tr} (K \rho)$ can be 
evaluated by XBM~\cite{kondo2022computationally} as described before. 

\begin{figure}[t]
    \centering
    \includegraphics[width=\columnwidth]{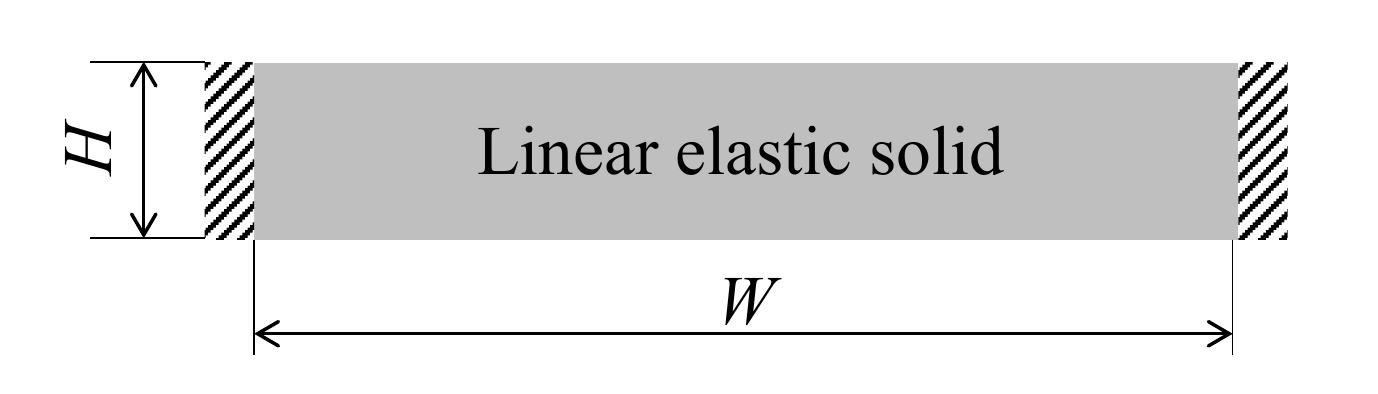}
    \caption{Simulation model of a linear elastic solid. Both sides are fixed.}
    \label{fig:model_solid}
\end{figure}

\begin{figure}[t]
    \centering
    \includegraphics[width=\columnwidth]{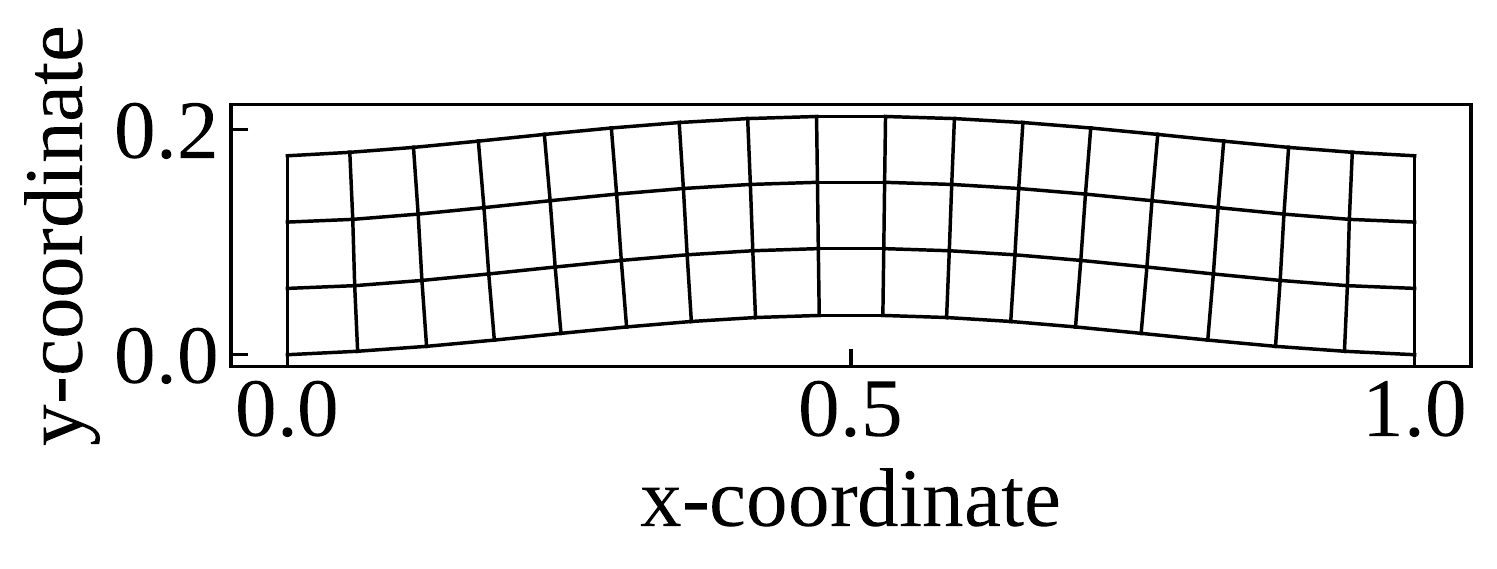}
    \caption{Deformation mode corresponding to the ground state.}
    \label{fig:eigenstate_solid}
\end{figure}

In this paper, we focus on the eigenvalue problem of a two-dimensional beam structure in the 
plain stress shown in Fig.~\ref{fig:model_solid}.
This beam structure is discretized using the first-order quadrilateral elements with 
$N_x \times N_y$ nodes where $N_x$ is the number of nodes along the $x$-axis and $N_y$ 
is the number of nodes along the $y$-axis. 
Since the both sides are fixed, the number of degrees of freedom is $2(N_x - 2)N_y$, 
which results in the required number of qubits 
$n = \lceil \log_2 (2(N_x - 2)N_y) \rceil = 1 + \lceil \log_2 ((N_x - 2)N_y) \rceil$. 
Here, we set $W=1$, $H=3/17$, $N_x = 18$, and $N_y = 4$, which leads to $n=7$. 
The nodes of the finite elements are numbered from bottom left to top right in 
Fig.~\ref{fig:eigenstate_solid}, and the binary representation of the node number is 
mapped to the computational basis of qubits. 
Note that the order of numbering and its mapping into qubits affect on the feasible solution generated by a PQC; this design problem is open for future work.

The Young's modulus and the Poisson ratio were respectively set to 200 GPa and 0.3. 
The material density $\varrho$ was set to $7850~\mathrm{kg/m^3}$. 
These settings are typical for iron.
Figure~\ref{fig:eigenstate_solid} illustrates the deformation mode corresponding to the 
ground state whose eigenvalue $\lambda$ is $2.55 \times 10^{7}~\mathrm{1/s^2}$ and the 
eigen-frequency $\omega / (2\pi) = \sqrt{\lambda} / (2\pi)$ is $8.04 \times 10^2~\mathrm{Hz}$.
Note that the solution, i.e., the displacement vector $\bs{u}$ of this problem, takes 
real values.

\subsubsection{Results and Discussion}

Here, we employed the alternating layered PQC in Fig.~\ref{fig:ALA_PQC} and used FQS, Fraxis and NFT. 
The 30 independent trials were performed for each method, and the complex space initialization was used for FQS.

\begin{figure}[t]
    \centering
    \includegraphics[width=\columnwidth]{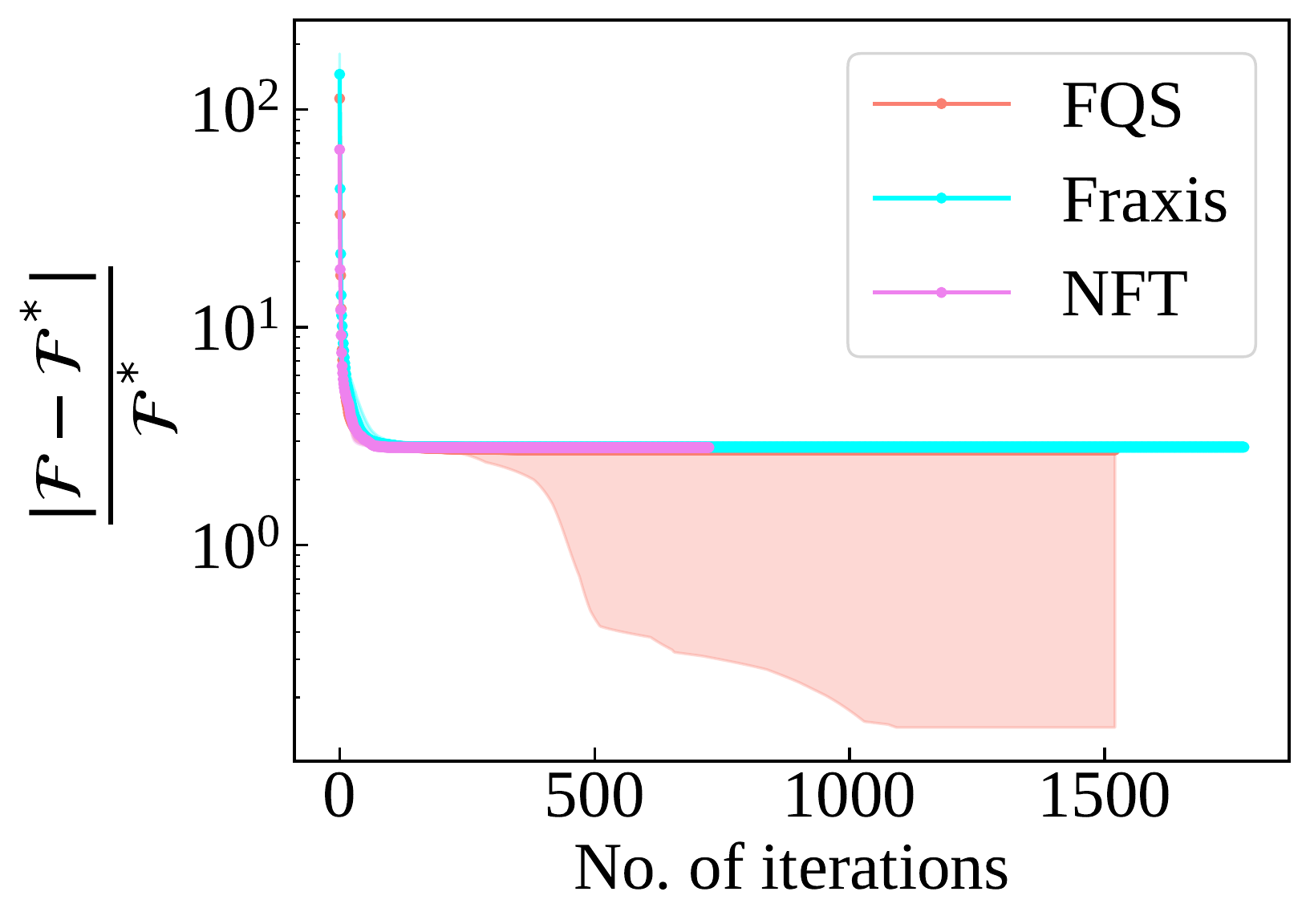}
    \caption{History of the relative error of the objective function value $\mathcal{F}$ to 
    the exact minimum value ${\mathcal F}^*$. 
    Plots represent the median value and shaded ranges represent ranges from 25- to 75- percentiles.}
    \label{fig:result_solid_ALA}
\end{figure}

Figure~\ref{fig:result_solid_ALA} shows histories of the relative error of the objective 
function value $\mathcal{F}$ to the exact minimum value ${\mathcal F}^*$, which were 
obtained from independent 30 optimizations. 
To calculate the median and percentiles, we took the same statistical treatment as described in Sec.~\ref{sec:poisson_result}. 
This figure indicates that half of trials of FQS and all trials of NFT and Fraxis got stuck 
in poor local optima where the relative errors are larger than 1. 
Nonetheless, the figure shows that the quarter of trials of FQS obtained much lower 
objective function values with the relative error of around 0.1. 
Thus, we again conclude that simultaneous and continuous optimization of angles and axes of gates is essential to fully exploiting the advantage to use the complex space. 
However, we actually observed that the many trials got stuck in poor local minima even 
when using this method (i.e., FQS with complex-valued initialization). 
We hope to address to construct a PQC structure which is suitable for such engineering 
problems in our future research.

\section{Conclusions} \label{sec:conclusions}

This paper proposed a VQA for calculating the minimum or maximum eigenvalue of a GEP characterized by Hermitian matrices, based on sequential quantum optimization techniques. 
First, we formulated the GEP as the minimization problem of the generalized Rayleigh quotient, 
which is reformulated in the fractional form of the expectations of two Hermitians as the objective function of the VQA.
We then showed that the objective function can be analytically minimized with respect to a single-qubit gate by solving GEP of a $4 \times 4$ matrix.
Second, we showed that an SLE characterized by a positive-definite matrix can be formulated as a GEP and thus be attacked using the proposed method. 
Finally, we demonstrated applications to two important engineering problems formulated with 
FEM; one is for solving an SLE derived from a Poisson equation, and the other is eigen-frequency 
analysis of a linear elastic solid. 
Through the demonstration, we found that a problem having a real-valued solution can be solved 
more effectively using quantum gates generating a complex-valued state vector. 
There are several open problems including; how to appropriately assign the ordering of  
optimization of single-qubit gates in PQC; to solve the eigen-frequency problem of a general 
elastic solid, how to suitably make numbering and mapping of finite element nodes into qubits; 
how to construct a PQC suitable for exploiting the complex Hilbert space to efficiently 
solve real-valued engineering problems. 
Also, as a future work we plan to determine an optimal number of shots for each quantum circuit for mitigating the sampling errors.

\appendix
\section{Minimum and maximum of Rayleigh quotient} \label{sec:rayleigh}
Here, we consider to find $\bs{w}$ that minimizes or maximizes the generalized Rayleigh quotient $R(\bs{w}; A, B)$ in Eq.~\eqref{eq:g_rayleigh}. 
The following theorem generalizes the previously known fact that holds when $A$ and $B$ are real-valued positive-definite matrices, e.g., see~\cite{yu2011kernel}.  

\begin{thm}
Let $A \in \mathbb{C}^{N \times N}$ and $B \in \mathbb{C}^{N \times N}$ be a Hermitian and a positive-definite Hermitian, respectively, and let $(\lambda_\mathrm{min}, \bs{v}_\mathrm{min})$ and $(\lambda_\mathrm{max}, \bs{v}_\mathrm{max})$ denote the minimum and maximum eigenvalue-eigenvector pairs of the generalized eigenvalue problem $A\bs{v} = \lambda B \bs{v}$, respectively.
Then, it holds that
\begin{align}
    &\lambda_\mathrm{min}\bs{w}^\dagger B \bs{w} \leq \bs{w}^\dagger A \bs{w} \leq \lambda_\mathrm{max}\bs{w}^\dagger B \bs{w} \quad \forall \bs{w} \in \mathbb{C}^N \label{eq:rayleigh_order} \\
    &\lambda_\mathrm{min} = \min_{\bs{w}} R(\bs{w}; A, B) \text{ s.t. } \bs{w} \neq \bs{0} \label{eq:rayleigh_min} \\
    &\lambda_\mathrm{max} = \max_{\bs{w}} R(\bs{w}; A, B) \text{ s.t. } \bs{w} \neq \bs{0}, \label{eq:rayleigh_max}
\end{align}
where $R(A, B; \bs{w})$ is the generalized Rayleigh quotient defined as
\begin{align}
    R(\bs{w}; A, B) := \dfrac{\bs{w}^\dagger A \bs{w}}{\bs{w}^\dagger B \bs{w}}.
\end{align}
\end{thm}

\begin{proof}
We first prove that Eq.~\eqref{eq:rayleigh_order} holds by contradiction.
Let us assume that $\bs{w}^\dagger A \bs{w} < \lambda_\mathrm{min}\bs{w}^\dagger B \bs{w}$.
Then, it holds that
\begin{align}
\bs{w}'^\dagger \left( B^{-1/2} A B^{-1/2} - \lambda_\mathrm{min} I \right) \bs{w}' < 0, \label{eq:leq_contradict}
\end{align}
where $\bs{w}'=B^{1/2}\bs{w}$, and it is ensured that $B^{1/2}$ and $B^{-1/2}$ exist because $B$ is positive-definite.
Eq.~\eqref{eq:leq_contradict} implies that the minimal eigenvalue of $B^{-1/2} A B^{-1/2}$ is lower than $\lambda_\mathrm{min}$.
However, this contradicts that $\lambda_\mathrm{min}$ is the minimal eigenvalue because the eigenvalues of $B^{-1/2} A B^{-1/2}$ correspond to those of $A\bs{v} = \lambda B \bs{v}$, and consequently $\lambda_\mathrm{min}\bs{w}^\dagger B \bs{w} \leq \bs{w}^\dagger A \bs{w}$.
Similarly, we obtain $\bs{w}^\dagger A \bs{w} \leq \lambda_\mathrm{max}\bs{w}^\dagger B \bs{w}$, which proves Eq.~\eqref{eq:rayleigh_order}.

Suppose $\bs{w} \neq \bs{0}$.
Since $B$ is positive definite, $\bs{w}^\dagger B \bs{w}$ is positive.
Dividing Eq.~\eqref{eq:rayleigh_order} by $\bs{w}^\dagger B \bs{w}$, we obtain
\begin{align}
    \lambda_\mathrm{min} \leq R(\bs{w}; A, B) \leq \lambda_\mathrm{max}.
\end{align}
That is, the generalized Rayleigh quotient is bounded by $\lambda_\mathrm{min}$ and $\lambda_\mathrm{max}$.
Furthermore, it holds that $\lambda_\mathrm{min} = R(\bs{v}_\mathrm{min}; A, B)$ and $\lambda_\mathrm{max} = R(\bs{v}_\mathrm{max}; A, B)$, which proved Eqs.~\eqref{eq:rayleigh_min} and \eqref{eq:rayleigh_max}.
\end{proof}

\section{The enlargement of matrices \textit{A} and \textit{B} for the \textit{n}-qubit system} \label{sec:padding}
Recall that we focus on a generalized eigenvalue problem (GEP) in Eq.~\ref{eq:g_eig} assuming $A \in \mathbb{C}^{N \times N}$ is a Hermitian matrix and $B \in \mathbb{C}^{N \times N}$ is a positive-definite Hermitian matrix.
The proposed method deals with the state vector of $n$-qubit system to represent the eigenvectors $\bs{v} \in \mathbb{C}^N$ of a GEP.
When $\log_2 N$ is not integer, we introduce enlarged matrices $\tilde{A} \in \mathbb{C}^{2^n \times 2^n}$ and $\tilde{B} \in \mathbb{C}^{2^n \times 2^n}$ to solve a GEP with the $n$-qubit system.
In the following, we consider two cases depending on the eigenvalue of interest.

\paragraph{When $\lambda_\mathrm{min} < 0$ or $\lambda_\mathrm{max} > 0$ is the eigenvalue of interest}
We introduce the enlarged matrices $\tilde{A} \in \mathbb{C}^{2^n \times 2^n}$ and $\tilde{B} \in \mathbb{C}^{2^n \times 2^n}$ defined as
\begin{align}
    &\tilde{A} := \begin{bmatrix}
    A & O_{N \times (2^n - N) } \\
    O_{N \times (2^n - N) }^\dagger & O_{(2^n - N) \times (2^n - N) }
    \end{bmatrix} \\
    &\tilde{B} := \begin{bmatrix}
    B & O_{N \times (2^n - N) } \\
    O_{N \times (2^n - N) }^\dagger & I_{(2^n - N) \times (2^n - N) }
    \end{bmatrix},
\end{align}
where $I_{ (2^n - N) \times (2^n - N) }$ represents $(2^n - N) \times (2^n - N)$ identity matrix, and $O_{(2^n - N) \times (2^n - N)}$ and $O_{N \times (2^n - N)}$ represents $(2^n - N) \times (2^n - N)$ and $N \times (2^n - N)$ zero matrices, respectively.
Considering the GEP expressed as
\begin{align} \label{eq:enlarged_gep}
    \tilde{A} \tilde{\bs{v}} = \lambda \tilde{B} \tilde{\bs{v}},
\end{align}
we obtain
\begin{align} \label{eq:enlarged_gep_sub}
\begin{cases}
    & A \bs{v} = \lambda B \bs{v} \\
    & \bs{0} = \lambda \bs{v}', 
\end{cases}
\end{align}
where $\tilde{\bs{v}} = (\bs{v}^\top, \bs{v}'^\top)^\top$.
Eq.~\eqref{eq:enlarged_gep_sub} means that the eigenvalues of the original GEP in Eq.~\eqref{eq:g_eig} are also eigenvalues of Eq.~\eqref{eq:enlarged_gep} with the corresponding eigenvectors $\tilde{\bs{v}}=(\bs{v}^\top, \bs{0}^\top)^\top$.
Additionally, $\lambda=0$ is also an eigenvalue of Eq.~\eqref{eq:enlarged_gep} where the corresponding eigenvector is $\tilde{\bs{v}}=(\bs{0}^\top, \bs{v}'^\top)^\top$.
When we are interested in the minimum eigenvalue $\lambda_\mathrm{min} < 0$ (the maximum eigenvalue $\lambda_\mathrm{max} > 0$) of the original GEP in Eq.~\eqref{eq:g_eig}, it is still the minimum (maximum) eigenvalue of the GEP in Eq.~\eqref{eq:enlarged_gep}.
Therefore, we can minimize (maximize) the generalized Rayleigh quotient $R(\tilde{\bs{w}}; \tilde{A}, \tilde{B})$ where $\tilde{\bs{w}} \in \mathbb{C}^{2^n}$ can be represented by a state vector of $n$-qubit system.

\paragraph{When $\lambda_\mathrm{min} > 0$ or $\lambda_\mathrm{max} < 0$ is the eigenvalue of interest}
We introduce the enlarged matrices $\tilde{A} \in \mathbb{C}^{2^n \times 2^n}$ and $\tilde{B} \in \mathbb{C}^{2^n \times 2^n}$ defined as
\begin{align}
    &\tilde{A} := \begin{bmatrix}
    A & O_{N \times (2^n - N) } \\
    O_{N \times (2^n - N) }^\dagger & I_{(2^n - N) \times (2^n - N) }
    \end{bmatrix} \\
    &\tilde{B} := \begin{bmatrix}
    B & O_{N \times (2^n - N) } \\
    O_{N \times (2^n - N) }^\dagger & \epsilon I_{(2^n - N) \times (2^n - N) }
    \end{bmatrix},
\end{align}
where $\epsilon$ is a constant.
Considering the GEP in Eq.~\eqref{eq:enlarged_gep}, we obtain
\begin{align} \label{eq:enlarged_gep_sub2}
\begin{cases}
    & A \bs{v} = \lambda B \bs{v} \\
    & \bs{v}' = \lambda \epsilon \bs{v}', 
\end{cases}
\end{align}
where $\tilde{\bs{v}} = (\bs{v}^\top, \bs{v}'^\top)^\top$.
Eq.~\eqref{eq:enlarged_gep_sub2} means that the eigenvalues of the original GEP in Eq.~\eqref{eq:g_eig} are also eigenvalues of Eq.~\eqref{eq:enlarged_gep_sub2} with the corresponding eigenvectors $\tilde{\bs{v}}=(\bs{v}^\top, \bs{0}^\top)^\top$.
Additionally, $\lambda=1 / \epsilon$ is also an eigenvalue of Eq.~\eqref{eq:enlarged_gep} where the corresponding eigenvector is $\tilde{\bs{v}}=(\bs{0}^\top, \bs{v}'^\top)^\top$.
When we are interested in the minimum eigenvalue $\lambda_\mathrm{min} > 0$, a sufficiently small constant $\epsilon > 0$ keeps $\lambda_\mathrm{min}$ to be the minimum eigenvalue of GEP in Eq.~\eqref{eq:enlarged_gep}.
In contrast, when we are interested in the maximum eigenvalue $\lambda_\mathrm{max} < 0$, a negative constant $\epsilon < 0$ with a sufficiently small absolute value keeps $\lambda_\mathrm{max}$ to be the maximum eigenvalue of GEP in Eq.~\eqref{eq:enlarged_gep}.
Therefore, we can minimize (maximize) the generalized Rayleigh quotient $R(\tilde{\bs{w}}; \tilde{A}, \tilde{B})$ where $\tilde{\bs{w}} \in \mathbb{C}^{2^n}$ can be represented by a state vector of $n$-qubit system.

\section{Asymptotic behavior of eigenvalues under sampling errors} \label{sec:sampling_error}
The proposed method solves a small size GEP in Eq.~\eqref{eq:eigenprob} to update parameters in a PQC.
Since the matrices $S(\rho', A')$ and $S(\rho', B')$ are constructed by expectation values of $A$ and $B$ calculated using several parameter sets, they will include the sampling errors, i.e. shot noises.
Here, we analyze the asymptotic behavior of eigenvalues of Eq.~\eqref{eq:eigenprob}.

Let $S_A$ and $S_B$ denote $S(\rho', A')$ and $S(\rho', B')$, respectively for simple notations.
Let $n_\mathrm{s}$ be the number of shots per individual quantum circuit for evaluating expectations of $A$ and $B$.
Since $S_A$ and $S_B$ are respectively constructed by the linear combination of expectation values of $A$ and $B$ calculated by several parameter sets, i.e. parameter configuration, the perturbations can be represented as
\begin{align}
    S_A &= S_A^{(0)} + \epsilon S_A^{(1)}, \label{eq:S_A_asymp} \\
    S_B &= S_B^{(0)} + \epsilon S_B^{(1)}, \label{eq:S_B_asymp}
\end{align}
where $\epsilon$ is $\mathcal{O}\left(1/\sqrt{n_\mathrm{s}}\right)$, the superscript $\cdot^{(0)}$ represents a quantity without any perturbation and $\cdot^{(1)}$ represents the first-order perturbation.
Assuming that the expectation values of observable can be estimated unbiasedly i.e. $\mathbb{E}[S_A] = S_A^{(0)}$ and $\mathbb{E}[S_B] = S_B^{(0)}$, we can suppose $\mathbb{E}[S_A^{(1)}] = \mathbb{E}[S_B^{(1)}] = O$ hold where $O$ is a zero matrix.
We now consider the asymptotic expansions of $\lambda_i$ and $\bs{p}_i$ are respectively given as
\begin{align}
    \lambda_i &= \lambda_i^{(0)} + \epsilon \lambda_i^{(1)} + \epsilon^2 \lambda_i^{(2)} + o(\epsilon^2), \label{eq:lam_asymp} \\
    \bs{p}_i &= \bs{p}_i^{(0)} + \epsilon \bs{p}_i^{(1)} + \epsilon^2 \bs{p}_i^{(2)} + o(\epsilon^2) \label{eq:q_asymp},
\end{align}
where the superscript $\cdot^{(2)}$ represents the second-order perturbation. 
Substituting Eqs.~\eqref{eq:S_A_asymp}--\eqref{eq:q_asymp} into Eq.~\eqref{eq:eigenprob}, we obtain the following equations for each order of $\epsilon$.

\paragraph{$\mathcal{O}(1)$ terms} These terms correspond to the GEP without any sampling errors and satisfy the equation given as
\begin{align}
    S_A^{(0)} \bs{p}_i^{(0)} = \lambda_i^{(0)} S_B^{(0)} \bs{p}_i^{(0)}. \label{eq:asymp_0}
\end{align}

\paragraph{$\mathcal{O}(\epsilon)$ terms} We obtain the equation given as
\begin{align}
    &S_A^{(0)} \bs{p}_i^{(1)} + S_A^{(1)} \bs{p}_i^{(0)} \nonumber \\
    &= \lambda_i^{(0)} S_B^{(0)} \bs{p}_i^{(1)} + \lambda_i^{(0)} S_B^{(1)} \bs{p}_i^{(0)} + \lambda_i^{(1)} S_B^{(0)} \bs{p}_i^{(0)}. \label{eq:asymp_1}
\end{align}
Multiplying the above equation on the left by $\bs{p}_i^{{(0)}\top}$ and using Eq.~\eqref{eq:asymp_0} and the normalized condition $\bs{p}_i^{{(0)}\top} S_B^{(0)} \bs{p}_i^{{(0)}} = 1$, we obtain
\begin{align}
    \lambda_i^{(1)} = \bs{p}_i^{{(0)}\top} \left( S_A^{(1)} - \lambda_i^{(0)} S_B^{(1)} \right) \bs{p}_i^{(0)}. \label{eq:lam_1}
\end{align}
Since $\mathbb{E}[S_A^{(1)}] = \mathbb{E}[S_B^{(1)}] = O$, the mean value of $\lambda_i^{(1)}$ is
\begin{align}
    \mathbb{E}[\lambda_i^{(1)}] = 0.
\end{align}

\paragraph{$\mathcal{O}(\epsilon^2)$ terms} We obtain the equation as follows:
\begin{align}
    &S_A^{(0)} \bs{p}_i^{(2)} + S_A^{(1)} \bs{p}_i^{(1)} \nonumber \\
    &= \lambda_i^{(0)} S_B^{(0)} \bs{p}_i^{(2)} + \lambda_i^{(0)} S_B^{(1)} \bs{p}_i^{(1)} \nonumber \\
    &\quad + \lambda_i^{(1)} S_B^{(0)} \bs{p}_i^{(1)} + \lambda_i^{(1)} S_B^{(1)} \bs{p}_i^{(0)} + \lambda_i^{(2)} S_B^{(0)} \bs{p}_i^{(0)}.
\end{align}
Multiplying the above equation on the left  by $\bs{p}_i^{{(0)}\top}$ and rearranging the equations using Eqs.~\eqref{eq:asymp_0} and \eqref{eq:asymp_1}, and the normalized condition $\bs{p}_i^{{(0)}\top} S_B^{(0)} \bs{p}_i^{{(0)}} = 1$, we obtain
\begin{align}
    \lambda_i^{(2)} &= -\lambda_i^{(1)}\bs{p}_i^{{(0)}\top} S_B^{(1)} \bs{p}_i^{{(0)}} \nonumber \\
    & \quad - \bs{p}_i^{{(1)}\top} \left( S_A^{(0)} - \lambda_i^{(0)}S_B^{(0)} \right) \bs{p}_i^{(1)}.
\end{align}
Substituting Eq.~\eqref{eq:lam_1}, we obtain the mean value of $\lambda_i^{(2)}$ as
\begin{align}
    \mathbb{E}[\lambda_i^{(2)}] &= \lambda_i^{(0)} \mathbb{E} \left[ \left( \bs{p}_i^{{(0)}\top}  S_B^{(1)} \bs{p}_i^{(0)} \right)^2 \right] \nonumber \\
    & \quad - \mathbb{E}\left[ \bs{p}_i^{{(1)}\top} \left( S_A^{(0)} - \lambda_i^{(0)}S_B^{(0)} \right) \bs{p}_i^{(1)} \right].
\end{align}
These two terms in the right-hand side are not zero in general, and thus these appear as a bias of eigenvalues in Eq.~\eqref{eq:eigenprob}, as follows:
\begin{align}
    \mathbb{E}\left[ \lambda_i \right] &= \lambda_i^{(0)} + \epsilon^2 \lambda_i^{(0)} \mathbb{E} \left[ \left( \bs{p}_i^{{(0)}\top}  S_B^{(1)} \bs{p}_i^{(0)} \right)^2 \right] \nonumber \\
    & \quad - \epsilon^2 \mathbb{E}\left[ \bs{p}_i^{{(1)}\top} \left( S_A^{(0)} - \lambda_i^{(0)}S_B^{(0)} \right) \bs{p}_i^{(1)} \right].
\end{align}
Therefore, the minimum (maximum) eigenvalue of Eq.~\eqref{eq:eigenprob} obtained with sampling errors is not an unbiased estimator of the optimal objective function value in Eq.~\eqref{eq:obj_q} achievable by updating one single-qubit gate of interest.
However, this bias vanishes asymptotically no slower than or equal to $\epsilon^2$, i.e. $\mathcal{O}(1/n_\mathrm{s})$.

Next, we estimate the actual objective function value after update of parameters using the perturbed eigenvector.
Here, we consider the minimization of the objective function, but the maximization is straightforward.
Substituting the perturbed eigenvector $\bs{p}_1$ into the objective function in Eq.~\eqref{eq:obj_q}, we obtain
\begin{widetext}
\begin{align} \label{eq:perturbed_cost_appendix}
    \mathcal{F}(\bs{p}_1) &= \dfrac{\bs{p}_1^\top S_A^{(0)} \bs{p}_1}{\bs{p}_1^\top S_B^{(0)} \bs{p}_1} \nonumber \\
    &\approx \dfrac{\bs{p}_1^{(0)\top} S_A^{(0)} \bs{p}_1^{(0)} + 2\epsilon \bs{p}_1^{(0)\top} S_A^{(0)} \bs{p}_1^{(1)} + \epsilon^2 (2\bs{p}_1^{(0)\top} S_A^{(0)} \bs{p}_1^{(2)} + \bs{p}_1^{(1)\top} S_A^{(0)} \bs{p}_1^{(1)}) }{\bs{p}_1^{(0)\top} S_B^{(0)} \bs{p}_1^{(0)} + 2\epsilon \bs{p}_1^{(0)\top} S_B^{(0)} \bs{p}_1^{(1)} + \epsilon^2 (2\bs{p}_1^{(0)\top} S_B^{(0)} \bs{p}_1^{(2)} + \bs{p}_1^{(1)\top} S_B^{(0)} \bs{p}_1^{(1)}) } \nonumber \\
    &\approx \lambda_1^{(0)} + \epsilon^2 \bs{p}_1^{(1)\top} \left( S_A^{(0)} - \lambda_1^{(0)} S_B^{(0)}\right) \bs{p}_1^{(1)}.
\end{align}
\end{widetext}
From the second to the third lines, we used the Taylor expansion with respect to $\epsilon$ up to the second order.
Let $\bs{p}^{(1)}$ be expanded using the $\{ \bs{p}_i^{(0)} \}$ as basis, as follows:
\begin{align}
    \bs{p}_1^{(1)} = \sum_j c_j \bs{p}_j^{(0)},
\end{align}
where $c_j$ is a coefficient.
Substituting this into Eq.~\eqref{eq:perturbed_cost_appendix} and considering the normalization condition $\bs{p}_i^{(0)} S_B^{(0)} \bs{p}_j^{(0)} = \delta_{ij}$, we obtain
\begin{align}
    \mathcal{F}(\bs{p}_1) &\approx \lambda_1^{(0)} + \epsilon^2 \bs{p}_1^{(1)\top} \left( S_A^{(0)} - \lambda_1^{(0)} S_B^{(0)}\right) \bs{p}_1^{(1)} \nonumber \\
    &= \lambda_1^{(0)} + \epsilon^2 \sum_{j} \| c_j \|^2 \left( \lambda_j^{(0)} - \lambda_1^{(0)}\right) \nonumber \\
    & \leq \lambda_1^{(0)} + \epsilon^2 \left( \lambda_\mathrm{max}^{(0)} - \lambda_1^{(0)} \right) \sum_{j} \| c_j \|^2 \nonumber \\
    &= \lambda_1^{(0)} + \epsilon^2 \left( \lambda_\mathrm{max}^{(0)} - \lambda_1^{(0)} \right) \bs{p}_1^{(1)\top} S_B^{(0)} \bs{p}_1^{(1)},
\end{align}
where $\lambda_\mathrm{max}$ is the maximum eigenvalue.
Therefore, the mean value of the objective function after parameter update is given as
\begin{align}
    &\mathbb{E}\left[ \mathcal{F}(\bs{p}_1) \right] \nonumber \\
    &\approx \lambda_1^{(0)} + \epsilon^2 \mathbb{E} \left[ \bs{p}_1^{(1)\top} \left( S_A^{(0)} - \lambda_1^{(0)} S_B^{(0)}\right) \bs{p}_1^{(1)} \right] \nonumber \\
    & \leq \lambda_1^{(0)} + \epsilon^2 \left( \lambda_\mathrm{max}^{(0)} - \lambda_1^{(0)} \right) \mathbb{E} \left[ \bs{p}_1^{(1)\top} S_B^{(0)} \bs{p}_1^{(1)} \right],
\end{align}
The second term is the gap between the ideal optimal value $\mathcal{F}(\bs{p}_1^{(0)})=\lambda_1^{(0)}$ and the mean value of the actual value $\mathcal{F}(\bs{p}_1)$, and vanishes asymptotically no slower than or equal to $\epsilon^2$.

\section{Overview of XBM} \label{sec:xbm}
To evaluate the expectation value of an Hermitian $H$, i.e., $\braket{H}=\tr{H\rho}$, this study employs the extended Bell measurement (XBM)~\cite{kondo2022computationally}.
Here, we briefly summarize the XBM.
Let $A\in\mathbb{C}^{2^n\times 2^n}$ be the arbitrary matrix and $A_{ij}$ the $(i,j)$-component of $A$.
In the XBM, the expectation value of $A$ with the quantum state vector $\ket{\psi}\in\mathbb{C}^{2^n}$ is written as
\begin{widetext}
\begin{equation}
\braket{\psi|A|\psi} = \sum^{2^n-1}_{i=0}A_{ii}|\langle i|\psi\rangle|^2
 + \sum_{l\in\{i\oplus j|A_{ij}\neq0\}\setminus\{0\}}\sum^{2^n-1}_{i=0}\left(a_\mathrm{Re}(A, l, i)|\langle i|M^{(l)}_\mathrm{Re}|\psi\rangle|^2+a_\mathrm{Im}(A, l, i)|\langle i|M^{(l)}_\mathrm{Im}|\psi\rangle|^2\right),
\end{equation}
\end{widetext}
where $a_\mathrm{Re}(A, l, i)$ and $a_\mathrm{Im}(A, l, i)$ are the function of $(A,l,i)$, and $M^{(l)}_{\rm Re}$ and $M^{(l)}_{\rm Im}$ are the measurement operators defined as 
\begin{equation}
\begin{cases}
\left(M^{(i\oplus j)}_\mathrm{Re}\right)^\dagger\Ket{i} = \dfrac{\Ket{i} + \Ket{j}}{\sqrt{2}} \\
\left(M^{(i\oplus j)}_\mathrm{Im}\right)^\dagger\Ket{i} = \dfrac{\Ket{i} + \iota \Ket{j}}{\sqrt{2}},
\end{cases}
\end{equation}
respectively.
Both measurement operators $M^{(l)}_{\rm Re}$ and $M^{(l)}_{\rm Im}$ can be expressed using at most one Hadamard gate, one phase gate, and $n-1$ CNOT gates.
It is known that when the bandwidth of $A$ is $k=\mathcal{O}(n^c)$ with a constant $c$, the number of groups for simultaneous measurement is $\mathcal{O}(nk)$~\cite{kondo2022computationally}.

\section{Matrix assembly in the finite element method~\cite{allaire2007numerical, hughes2012finite}} \label{sec:fem}
Here, we briefly explain the matrix assembly in the finite element method~\cite{allaire2007numerical, hughes2012finite} to construct the stiffness matrix $K$ for both the Poisson equation and the linear elastic solid, and the mass matrix $M$ for the linear elastic solid.
The main procedure is (1) deriving the weak form of the PDE to be solved, and (2) discretizing the weak form by introducing the shape function.

\subsection{Assembly for the Poisson equation} \label{sec:fem_poisson}
First, let us derive the weak form of the Poisson equation in Eq.~\eqref{eq:poisson}.
Let $\tilde{u} \in H^1_0(\Omega)$ be an arbitrary function in the space $H^1_0(\Omega)$ defined as
\begin{align}
    H^1_0(\Omega) := \{ \tilde{u} \in H^1(\Omega) ~|~ \tilde{u}(\bs{x})=0 ~\text{on}~ \partial \Omega  \},
\end{align}
where $H^1(\Omega)$ is the Sobolev space.
Multiplying Eq.~\eqref{eq:poisson} by the test function $\tilde{u} (\bs{x})$, and integrating both sides over the domain $\Omega$, we obtain
\begin{align} \label{eq:int_poisson}
    -\int_\Omega \tilde{u}(\bs{x}) \nabla^2 u(\bs{x}) \diff \Omega = \int_\Omega \tilde{u}(\bs{x}) f(\bs{x}) \diff \Omega.
\end{align}
Applying the integration by parts and the Gauss's theorem to Eq.~\eqref{eq:int_poisson}, we derive the weak form of the Poisson equation as
\begin{align} \label{eq:weak_poisson}
    &\text{find} \quad u \in H^1_0(\Omega) \nm \\
    &\text{s.t.} \int_\Omega \nabla \tilde{u}(\bs{x}) \cdot \nabla u(\bs{x}) \diff \Omega = \int_\Omega \tilde{u}(\bs{x}) f(\bs{x}) \diff \Omega ~ \forall \tilde{u} \in H^1_0(\Omega).
\end{align}

Next, we approximate the solution of the weak form by introducing the set of shape functions $ \{ \phi_j \}_{j=1}^N$ as $u(\bs{x}) \approx \sum_j \phi_j (\bs{x})u_j$ and we set $\bs{U} := [ u_1, \ldots u_N ]^\top$.
Substituting $u(\bs{x}) \approx \sum_j \phi_j (\bs{x})u_j$ into the weak form in Eq.~\eqref{eq:weak_poisson} and taking $\tilde{u} = \phi_i$, we obtain
\begin{align}
\sum_{j=1}^N u_j \int_\Omega \nabla \phi_i(\bs{x}) \cdot \nabla \phi_j(\bs{x}) \diff \Omega = \int_\Omega \phi_i (\bs{x}) f(\bs{x}) \diff \Omega.
\end{align}
Considering this equation in $i = 1, \ldots, N$ yields a linear system of the form
\begin{align}
    K \bs{U} = \bs{F},
\end{align}
where
\begin{align}
    (K)_{i j} &:= \int_\Omega \nabla \phi_i(\bs{x}) \cdot \nabla \phi_j(\bs{x}) \diff \Omega \label{eq:poisson_K} \\
    (\bs{F})_{i} &:= \int_\Omega \phi_i (\bs{x}) f(\bs{x}) \diff \Omega. \label{eq:poisson_f}
\end{align}

Focusing on the case of one dimension, we consider the domain $\Omega=(0, 1)$ and introduce the uniform mesh with nodes $0 < h=x_1 < x_2 < \ldots < x_N = 1-h < 1$ where $x_{j+1} - x_j = h$ for all $j=1, \ldots N-1$.
The first-order element we used employs the shape function $\phi_j$ defined as
\begin{align} \label{eq:poisson_shape}
    \phi_j (x) = 
    \begin{cases}
        1 - \dfrac{|x - x_j|}{h} & \text{if } |x - x_j| \leq h \\
        0 & \text{if } |x - x_j| > h.
    \end{cases}
\end{align}
Substituting Eq.~\eqref{eq:poisson_shape} into Eq.~\eqref{eq:poisson_K} followed by a simple calculation shows that the striffness matrix $K$ is tridiagonal matrix written as
\begin{align}
    (K)_{i j} = \dfrac{1}{h} \begin{bmatrix}
    2  & -1 & 0  & \cdots & 0 \\
    -1 & 2  & -1 & \cdots &  \\
      & \ddots & \ddots  & \ddots & \\
      &  & -1  & 2 & -1 \\
    0  & \cdots & 0 & -1 & 2
    \end{bmatrix}.
\end{align}
To obtain the vector $\bs{F}$, we can use a certain quadrature formula to calculate the integration of the right hand side of Eq.~\eqref{eq:poisson_f}.
In the present study, we take $(\bs{F})_j = f(x_j)$, for simplicity.

\subsection{Assembly for the eigenvalue problem of a linear elastic solid} \label{sec:fem_solid}
We proceed in a similar manner to the discretization of the Poisson equation in Sec.~\ref{sec:fem_poisson} to obtain the stiffness matrix $K$ and the mass matrix $M$ for the eigenfrequency analysis of a linear elastic solid by the finite element method~\cite{allaire2007numerical, hughes2012finite}.
By a similar procedure to the case of the Poisson equation, we obtain the weak form of Eq.~\eqref{eq:solid} as
\begin{align} \label{eq:weak_solid}
    &\text{find} \quad \bs{u} \in \mathcal{U} \nm \\
    &\text{s.t.} \int_\Omega \tilde{\bs{\varepsilon}} : \bs{C} : \bs{\varepsilon} \diff \Omega = \int_\Omega \varrho \omega^2 \tilde{\bs{u}} \cdot \bs{u} \diff \Omega \quad \forall \tilde{\bs{u}} \in \mathcal{U},
\end{align}
where $\varrho$ is the density of the solid,
\begin{align}
    \bs{\varepsilon}(\bs{x}) &= \dfrac{1}{2} \left( \nabla \bs{u} + \nabla \bs{u}^\top \right) \\
    \tilde{\bs{\varepsilon}}(\bs{x}) &= \dfrac{1}{2} \left( \nabla \tilde{\bs{u}} + \nabla \tilde{\bs{u}}^\top \right),
\end{align}
and $\mathcal{U}$ is the Sobolev space defined as
\begin{align}
    \mathcal{U} := \{ \bs{u} \in H^1(\Omega)^m ~|~ \bs{u}(\bs{x}) = 0  ~\text{on}~ \Gamma_\text{D} \}. 
\end{align}
$\Gamma_\text{D} \subset \partial \Omega$ is the boundary on which the displacement $\bs{u}$ is fixed.

We now consider to approximate the solution of the weak from by using the shape functions, i.e., $\bs{u}(\bs{x}) \approx \sum_j \phi_j (\bs{x}) \bs{u}_j$ and we set $\bs{U} := [\bs{u}_1^\top, \ldots, \bs{u}_N^\top]^\top$ which has the dimension of $Nm$ since the solution $\bs{u}$ is a vector-valued function that takes values in $\mathbb{R}^m$.
That is, the $(m(j-1)+l)$-th component of $\bs{U}$ represents the displacement at $j$-th node along the $l$-th axis, i.e., $(\bs{u}_j)_l$.
Substituting $\bs{u} \approx \sum_j \phi_j (\bs{x}) \bs{u}_j$ and Eq.~\eqref{eq:elastic_tensor} into Eq.~\eqref{eq:weak_solid} and taking $(\tilde{\bs{u}})_{k'} = \delta_{k k'} \phi_i$, we obtain
\begin{widetext}
\begin{align} \label{eq:weak_solid_approx}
& \sum_{j, l} (\bs{u}_j)_l \int_\Omega \left( \dfrac{E \nu}{(1 + \nu)(1 - 2\nu)} \phi_{i, k}(\bs{x}) \phi_{j, l}(\bs{x}) + \dfrac{E}{2(1 + \nu)} \delta_{kl} \nabla \phi_i (\bs{x}) \cdot \nabla \phi_j (\bs{x}) + \dfrac{E}{2(1 + \nu)} \phi_{j,k} \phi_{i, l} \right) \diff \Omega \nm \\
&= \lambda \sum_{j, l} (\bs{u}_j)_l \int_\Omega \varrho \delta_{kl} \phi_i(\bs{x}) \phi_j (\bs{x}) \diff \Omega,
\end{align}
\end{widetext}
where $\lambda = \omega^2$ and $\phi_{\alpha, \beta}$ represents $\partial \phi_\alpha / \partial x_\beta$.
Considering Eq.~\eqref{eq:weak_solid_approx} in $i=1, \ldots, N$ and $k=1 \ldots, m$ yields a generalized eigenvalue problem of the form
\begin{align}
    K \bs{U} = \lambda M \bs{U},
\end{align}
where
\begin{widetext}
\begin{align}
    (K)_{i' j'} &:= \int_\Omega \left( \dfrac{E \nu}{(1 + \nu)(1 - 2\nu)} \phi_{i, k}(\bs{x}) \phi_{j, l}(\bs{x}) + \dfrac{E}{2(1 + \nu)} \delta_{kl} \nabla \phi_i (\bs{x}) \cdot \nabla \phi_j (\bs{x}) + \dfrac{E}{2(1 + \nu)} \phi_{j,k} \phi_{i, l} \right) \diff \Omega \\
    (M)_{i' j'} &:= \int_\Omega \varrho \delta_{kl} \phi_i(\bs{x}) \phi_j (\bs{x}) \diff \Omega,
\end{align}
\end{widetext}
using the notation $i' := m(i-1)+k$ and $j' := m(j-1)+l$.

Focusing on the case of two dimensions, we consider the domain $\Omega=(0, W) \times (0, H)$ and introduce the grid mesh with nodes $(x_{j_x}, y_{j_y})$ defined by the equally-spaced points along each axis, $0 = x_1 < x_2 < \ldots < x_{N_x} = W$ and $0 = y_1 < y_2 < \ldots < y_{N_y} = H$ where $x_{j_x+1} - x_{j_x} = h_x$ for all ${j_x}=1, \ldots N_x-1$ and $y_{j_y+1} - y_{j_y} = h_y$ for all $j_y=1, \ldots N_y-1$, respectively.
We used the first-order element which employs the shape function $\phi_j$ defined as
\begin{widetext}
\begin{align} \label{eq:solid_shape}
    \phi_j (x, y) = 
    \begin{cases}
        \left( 1 - \dfrac{|x - x_{j_x}|}{h_x} \right) \left( 1 - \dfrac{|y - y_{j_y}|}{h_y} \right) & \text{if } |x - x_{j_x}| \leq h_x \text{ and } |y - y_{j_y}| \leq h_y \\
        0 & \text{otherwise},
    \end{cases}
\end{align}
\end{widetext}
associating $j$ with $j_x$ and $j_y$ as $j = N_y(j_x - 1) + j_y$.
We used the Gauss quadrature rule for calculating the integration.

\section{Optimal parameter configuration for generalized eigenvalue problem} \label{sec:param_config}
In the following, we consider the minimization of the objective function, the minimum eigenvalue of Eq.~\eqref{eq:eigenprob}. 
The analysis for the maximization problems is straightforward.
Let $P := \{ \tilde{\bs{q}}_i \}_{i=1}^r$ be the parameter configuration, the set of $r$ parameters, and $\bs{h}_A := (h_{A1}, \cdots h_{Ar})^\top$ and $\bs{h}_B := (h_{B1}, \cdots h_{Br})^\top$ be the expectation values of $A$ and $B$ through the parameter in the configuration.
Typically, $r = 10$ for FQS, $r=6$ for Fraxis, and $r=3$ for NFT.
Let $S_A$ and $S_B$ denote $S(\rho', A')$ and $S(\rho', B')$, respectively for simple notations.
Since $S_A$ and $S_B$ are respectively constructed by applying the linear transformation depending on the parameter configuration to $\bs{h}_A$ and $\bs{h}_B$~\cite{endo2023optimal}, $S_A$ and $S_B$ are represented as
\begin{align}
    S_A &= T(\bs{h}_A | P), \\
    S_B &= T(\bs{h}_B | P),
\end{align}
where $T$ is the linear transformation depending of the parameter configuration $P$.
Now, we assume that the expectation values are perturbed due to sampling errors around $\bs{h}_A$ and $\bs{h}_B$ by uncorrelated Gaussian noises $\bs{\delta}_A \sim \mathcal{N}(\bs{0}, \sigma_A^2 / n_\mathrm{s} I )$ and $\bs{\delta}_B \sim \mathcal{N}(\bs{0}, \sigma_B^2 / n_\mathrm{s} I )$, respectively, where $\sigma_A^2$ and $\sigma_B^2$ are the variances of estimating each entry of $\bs{h}_A$ and $\bs{h}_B$, respectively, and $I$ is the $r \times r$ identity matrix.
Note that the variances of each entry of $\bs{h}_A$ and $\bs{h}_B$ are not in general the same since they depend on parameters of PQC.
However, we herein assume the variances are the same to discuss the parameter configuration independent from particular parameter values of PQC.
The assumption of the uncorrelated Gaussian noises correspond to the assumption that the estimation of the expectation values of $A$ and $B$ is unbiased, i.e. $\mathbb{E}[\bs{\delta}_{A}] = \mathbb{E}[\bs{\delta}_{B}] = \bs{0}$ and that the measurement outcomes of quantum circuits are independent, i.e. $\mathbb{E}[\delta_{Ai} \delta_{Aj}] = \mathbb{E}[\delta_{Bi} \delta_{Bj}] = \mathbb{E}[\delta_{Ai}\delta_{Bj}] = 0$ for $i, j (\neq i) \in \{1, \cdots, r \}$.
Under sampling errors, $S_A$ and $S_B$ are now represented as
\begin{align}
    S_A &= T(\bs{h}_A | P) + T(\bs{\delta}_A | P), \\
    S_B &= T(\bs{h}_B | P) + T(\bs{\delta}_B | P).
\end{align}
That is, $T(\bs{h}_A | P)$ and $T(\bs{h}_B | P)$ correspond to $S(\rho', A')^{(0)}$ and $S(\rho', B')^{(0)}$, respectively, and $T(\bs{\delta}_A | P)$ and $T(\bs{\delta}_B | P)$ correspond to $\epsilon S(\rho', A')^{(1)}$ and $\epsilon S(\rho', B')^{(1)}$, respectively.
To derive the optimal parameter configuration, we focus on the variance of the minimum eigenvalue in Eq.~\eqref{eq:eigenprob}.
The smaller the variance of the eigenvalue is, the more accurately the objective function value is estimated under the sampling errors.
Based on the first-order perturbation, the variance of the minimum eigenvalue is given as
\begin{align}
    \mathrm{Var}[\lambda_1] = \mathrm{Var}\left[ \bs{p}_i^{{(0)}\top} T ( \bs{\delta}_A - \lambda_1^{(0)} \bs{\delta}_B | P ) \bs{p}_i^{(0)} \right]. 
\end{align}
To minimize $\mathrm{Var}[\lambda_1]$ with respect to $P$, we can use the same analysis as that in Ref~\cite{endo2023optimal}, assuming that $\bs{\delta}_A - \lambda_1^{(0)} \bs{\delta}_B$ follows the uncorrelated Gaussian distribution, which is actually satisfied from the assumption of $\bs{\delta}_A$ and $\bs{\delta}_B$.
Thus, we can evaluate Eq.~\eqref{eq:S_components} using the optimal parameter configuration, which is the set of parameters aligned in a symmetric manner, to mitigate the sampling errors.

\section*{Acknowledgement}
This work was supported by the MEXT Quantum Leap Flagship Program Grant Number JPMXS0118067285 and JPMXS0120319794. 
H.C.W. was supported by JSPS Grant Numbers 20K03885. 
N.Y. was supported by JSPS KAKENHI Grant Number 20H05966.

\bibliographystyle{unsrt}
\bibliography{ref}

\begin{thebibliography}{10}

\bibitem{boffi2010finite}
Daniele Boffi.
\newblock Finite element approximation of eigenvalue problems.
\newblock {\em Acta Numerica}, 19:1--120, 2010.

\bibitem{yu2011kernel}
Shi Yu, L{\'e}on-Charles Tranchevent, Bart De~Moor, and Yves Moreau.
\newblock Kernel-based data fusion for machine learning.
\newblock {\em Studies in Computational Intelligence: Springer Berlin
  Heidelberg}, 2011.

\bibitem{ford1974generalized}
Brian Ford and George Hall.
\newblock The generalized eigenvalue problem in quantum chemistry.
\newblock {\em Computer Physics Communications}, 8(5):337--348, 1974.

\bibitem{hughes2012finite}
T.~J. Hughes.
\newblock {\em The finite element method: linear static and dynamic finite
  element analysis}.
\newblock Courier Corporation, 2012.

\bibitem{Klawonn2015}
Axel Klawonn, Martin Lanser, and Oliver Rheinbach.
\newblock Toward extremely scalable nonlinear domain decomposition methods for
  elliptic partial differential equations.
\newblock {\em SIAM Journal on Scientific Computing}, 37(6):C667--C696, 2015.

\bibitem{Toivanen2018}
Jari Toivanen, Philip Avery, and Charbel Farhat.
\newblock A multilevel {FETI-DP} method and its performance for problems with
  billions of degrees of freedom.
\newblock {\em International Journal for Numerical Methods in Engineering},
  116(10-11):661--682, 2018.

\bibitem{aspuru2005simulated}
Al{\'a}n Aspuru-Guzik, Anthony~D Dutoi, Peter~J Love, and Martin Head-Gordon.
\newblock Simulated quantum computation of molecular energies.
\newblock {\em Science}, 309(5741):1704--1707, 2005.

\bibitem{o2016scalable}
Peter~JJ O’Malley, Ryan Babbush, Ian~D Kivlichan, Jonathan Romero, Jarrod~R
  McClean, Rami Barends, Julian Kelly, Pedram Roushan, Andrew Tranter, Nan
  Ding, et~al.
\newblock Scalable quantum simulation of molecular energies.
\newblock {\em Physical Review X}, 6(3):031007, 2016.

\bibitem{parker2020quantum}
Jeffrey~B Parker and Ilon Joseph.
\newblock Quantum phase estimation for a class of generalized eigenvalue
  problems.
\newblock {\em Physical Review A}, 102(2):022422, 2020.

\bibitem{peruzzo2014NatCom}
Alberto Peruzzo, Jarrod McClean, Peter Shadbolt, Man-Hong Yung, Xiao-Qi Zhou,
  Peter~J. Love, Alán Aspuru-Guzik, and Jeremy~L. O’Brien.
\newblock A variational eigenvalue solver on a photonic quantum processor.
\newblock {\em Nat. Commun.}, 5(1):4213, 2014.

\bibitem{kandala2017Nat}
Abhinav Kandala, Antonio Mezzacapo, Kristan Temme, Maika Takita, Markus Brink,
  Jerry~M. Chow, and Jay~M. Gambetta.
\newblock Hardware-efficient variational quantum eigensolver for small
  molecules and quantum magnets.
\newblock {\em Nature}, 549:242--246, 2017.

\bibitem{li2019variational}
Yifan Li, Jiaqi Hu, Xiao-Ming Zhang, Zhigang Song, and Man-Hong Yung.
\newblock Variational quantum simulation for quantum chemistry.
\newblock {\em Advanced Theory and Simulations}, 2(4):1800182, 2019.

\bibitem{zhang2021shallow}
Feng Zhang, Niladri Gomes, Noah~F Berthusen, Peter~P Orth, Cai-Zhuang Wang,
  Kai-Ming Ho, and Yong-Xin Yao.
\newblock Shallow-circuit variational quantum eigensolver based on
  symmetry-inspired hilbert space partitioning for quantum chemical
  calculations.
\newblock {\em Physical Review Research}, 3(1):013039, 2021.

\bibitem{cerezo2021variational}
Marco Cerezo, Andrew Arrasmith, Ryan Babbush, Simon~C Benjamin, Suguru Endo,
  Keisuke Fujii, Jarrod~R McClean, Kosuke Mitarai, Xiao Yuan, Lukasz Cincio,
  et~al.
\newblock Variational quantum algorithms.
\newblock {\em Nature Reviews Physics}, 3(9):625--644, 2021.

\bibitem{TILLY20221}
Jules Tilly, Hongxiang Chen, Shuxiang Cao, Dario Picozzi, Kanav Setia, Ying Li,
  Edward Grant, Leonard Wossnig, Ivan Rungger, George~H. Booth, and Jonathan
  Tennyson.
\newblock The variational quantum eigensolver: A review of methods and best
  practices.
\newblock {\em Physics Reports}, 986:1--128, 2022.

\bibitem{gomes2021adaptive}
Niladri Gomes, Anirban Mukherjee, Feng Zhang, Thomas Iadecola, Cai-Zhuang Wang,
  Kai-Ming Ho, Peter~P Orth, and Yong-Xin Yao.
\newblock Adaptive variational quantum imaginary time evolution approach for
  ground state preparation.
\newblock {\em Advanced Quantum Technologies}, 4(12):2100114, 2021.

\bibitem{higgott2019variational}
Oscar Higgott, Daochen Wang, and Stephen Brierley.
\newblock Variational quantum computation of excited states.
\newblock {\em Quantum}, 3:156, 2019.

\bibitem{nakanishi2019subspace}
Ken~M Nakanishi, Kosuke Mitarai, and Keisuke Fujii.
\newblock Subspace-search variational quantum eigensolver for excited states.
\newblock {\em Physical Review Research}, 1(3):033062, 2019.

\bibitem{gocho2023excited}
Shigeki Gocho, Hajime Nakamura, Shu Kanno, Qi~Gao, Takao Kobayashi, Taichi
  Inagaki, and Miho Hatanaka.
\newblock Excited state calculations using variational quantum eigensolver with
  spin-restricted ans{\"a}tze and automatically-adjusted constraints.
\newblock {\em npj Computational Materials}, 9(1):13, 2023.

\bibitem{hirai2023excited}
Hirotoshi Hirai.
\newblock Excited-state molecular dynamics simulation based on variational
  quantum algorithms.
\newblock {\em Chemical Physics Letters}, 816:140404, 2023.

\bibitem{benedetti2021}
Marcello Benedetti, Mattia Fiorentini, and Michael Lubasch.
\newblock Hardware-efficient variational quantum algorithms for time evolution.
\newblock {\em Physical Review Research}, 3(3), Jul 2021.

\bibitem{wada2022}
Kaito Wada, Rudy Raymond, Yu-ya Ohnishi, Eriko Kaminishi, Michihiko Sugawara,
  Naoki Yamamoto, and Hiroshi~C Watanabe.
\newblock Simulating time evolution with fully optimized single-qubit gates on
  parametrized quantum circuits.
\newblock {\em Physical Review A}, 105(6):062421, 2022.

\bibitem{liu2021variational}
Hai-Ling Liu, Yu-Sen Wu, Lin-Chun Wan, Shi-Jie Pan, Su-Juan Qin, Fei Gao, and
  Qiao-Yan Wen.
\newblock Variational quantum algorithm for the poisson equation.
\newblock {\em Physical Review A}, 104(2):022418, 2021.

\bibitem{sato2021variational}
Yuki Sato, Ruho Kondo, Satoshi Koide, Hideki Takamatsu, and Nobuyuki Imoto.
\newblock Variational quantum algorithm based on the minimum potential energy
  for solving the poisson equation.
\newblock {\em Physical Review A}, 104(5):052409, 2021.

\bibitem{demirdjian2022variational}
Reuben Demirdjian, Daniel Gunlycke, Carolyn~A Reynolds, James~D Doyle, and
  Sergio Tafur.
\newblock Variational quantum solutions to the advection--diffusion equation
  for applications in fluid dynamics.
\newblock {\em Quantum Information Processing}, 21(9):1--13, 2022.

\bibitem{bravo2019variational}
Carlos Bravo-Prieto, Ryan LaRose, Marco Cerezo, Yigit Subasi, Lukasz Cincio,
  and Patrick~J Coles.
\newblock Variational quantum linear solver.
\newblock {\em arXiv preprint arXiv:1909.05820}, 2019.

\bibitem{xu2021variational}
Xiaosi Xu, Jinzhao Sun, Suguru Endo, Ying Li, Simon~C Benjamin, and Xiao Yuan.
\newblock Variational algorithms for linear algebra.
\newblock {\em Science Bulletin}, 66(21):2181--2188, 2021.

\bibitem{larose2019variational}
Ryan LaRose, Arkin Tikku, {\'E}tude O’Neel-Judy, Lukasz Cincio, and Patrick~J
  Coles.
\newblock Variational quantum state diagonalization.
\newblock {\em npj Quantum Information}, 5(1):1--10, 2019.

\bibitem{cerezo2022variational}
M~Cerezo, Kunal Sharma, Andrew Arrasmith, and Patrick~J Coles.
\newblock Variational quantum state eigensolver.
\newblock {\em npj Quantum Information}, 8(1):1--11, 2022.

\bibitem{liang2020variational}
Jin-Min Liang, Shu-Qian Shen, Ming Li, and Lei Li.
\newblock Variational quantum algorithms for dimensionality reduction and
  classification.
\newblock {\em Physical Review A}, 101(3):032323, 2020.

\bibitem{liang2022quantum}
Jin-Min Liang, Shu-Qian Shen, Ming Li, and Shao-Ming Fei.
\newblock Quantum algorithms for the generalized eigenvalue problem.
\newblock {\em Quantum Information Processing}, 21:1--22, 2022.

\bibitem{nakanishi2020}
Ken~M Nakanishi, Keisuke Fujii, and Synge Todo.
\newblock Sequential minimal optimization for quantum-classical hybrid
  algorithms.
\newblock {\em Physical Review Research}, 2(4):043158, 2020.

\bibitem{ostaszewski2021}
Mateusz Ostaszewski, Edward Grant, and Marcello Benedetti.
\newblock Structure optimization for parameterized quantum circuits.
\newblock {\em Quantum}, 5:391, 2021.

\bibitem{watanabe2021}
Hiroshi~C Watanabe, Rudy Raymond, Yu-Ya Ohnishi, Eriko Kaminishi, and Michihiko
  Sugawara.
\newblock Optimizing parameterized quantum circuits with free-axis selection.
\newblock In {\em 2021 IEEE International Conference on Quantum Computing and
  Engineering (QCE)}, pages 100--111. IEEE, 2021.

\bibitem{wada2022full}
Kaito Wada, Rudy Raymond, Yuki Sato, and Hiroshi~C Watanabe.
\newblock Sequential optimal selection of a single-qubit gate and its relation
  to barren plateau in parameterized quantum circuits.
\newblock {\em arXiv preprint arXiv:2209.08535v4}, 2023.

\bibitem{evans2010partial}
Lawrence~C Evans.
\newblock {\em Partial differential equations}, volume~19.
\newblock American Mathematical Soc., 2010.

\bibitem{deisenroth2020mathematics}
Marc~Peter Deisenroth, A~Aldo Faisal, and Cheng~Soon Ong.
\newblock {\em Mathematics for machine learning}.
\newblock Cambridge University Press, 2020.

\bibitem{aggarwal2020linear}
Charu~C Aggarwal, Lagerstrom-Fife Aggarwal, and Lagerstrom-Fife.
\newblock {\em Linear algebra and optimization for machine learning}, volume
  156.
\newblock Springer, 2020.

\bibitem{liu2022application}
YY~Liu, Zhen Chen, Chang Shu, Siou-Chye Chew, Boo~Cheong Khoo, Xiang Zhao, and
  YD~Cui.
\newblock Application of a variational hybrid quantum-classical algorithm to
  heat conduction equation and analysis of time complexity.
\newblock {\em Physics of Fluids}, 34(11):117121, 2022.

\bibitem{allaire2007numerical}
Gr{\'e}goire Allaire.
\newblock {\em Numerical analysis and optimization: an introduction to
  mathematical modelling and numerical simulation}.
\newblock OUP Oxford, 2007.

\bibitem{endo2023optimal}
Katsuhiro Endo, Yuki Sato, Rudy Raymond, Kaito Wada, Naoki Yamamoto, and
  Hiroshi~C Watanabe.
\newblock Optimal parameter configurations for sequential optimization of
  variational quantum eigensolver.
\newblock {\em arXiv preprint arXiv:2303.07082}, 2023.

\bibitem{cardano1560}
Girolamo Cardano.
\newblock {\em Ars magna or The rules of algebra}.
\newblock New York: Dover, 1993.

\bibitem{harrow2009quantum}
Aram~W Harrow, Avinatan Hassidim, and Seth Lloyd.
\newblock Quantum algorithm for linear systems of equations.
\newblock {\em Physical review letters}, 103(15):150502, 2009.

\bibitem{kondo2022computationally}
Ruho Kondo, Yuki Sato, Satoshi Koide, Seiji Kajita, and Hideki Takamatsu.
\newblock Computationally efficient quantum expectation with extended bell
  measurements.
\newblock {\em Quantum}, 6:688, 2022.

\bibitem{ribeiro2013finite}
Diogo Ribeiro, Rui Cal{\c{c}}ada, Raimundo Delgado, Maik Brehm, and Volkmar
  Zabel.
\newblock Finite-element model calibration of a railway vehicle based on
  experimental modal parameters.
\newblock {\em Vehicle System Dynamics}, 51(6):821--856, 2013.

\bibitem{muhammad2020finite}
AISHA Muhammad, MAH Ali, and IH~Shanono.
\newblock Finite element analysis of a connecting rod in ansys: An overview.
\newblock In {\em IOP Conference Series: Materials Science and Engineering},
  volume 736, page 022119. IOP Publishing, 2020.

\bibitem{belhocine2020thermomechanical}
Ali Belhocine and Oday~Ibraheem Abdullah.
\newblock Thermomechanical model for the analysis of disc brake using the
  finite element method in frictional contact.
\newblock {\em Multiscale Science and Engineering}, 2:27--41, 2020.

\bibitem{Zhang2021low}
Xiao-Ming Zhang, Man-Hong Yung, and Xiao Yuan.
\newblock Low-depth quantum state preparation.
\newblock {\em Physical Review Research}, 3(4):043200, 2021.

\bibitem{nakaji2022approximate}
Kouhei Nakaji, Shumpei Uno, Yohichi Suzuki, Rudy Raymond, Tamiya Onodera,
  Tomoki Tanaka, Hiroyuki Tezuka, Naoki Mitsuda, and Naoki Yamamoto.
\newblock Approximate amplitude encoding in shallow parameterized quantum
  circuits and its application to financial market indicators.
\newblock {\em Physical Review Research}, 4(2):023136, 2022.

\bibitem{Qiskit}
H{\'e}ctor Abraham~\textit{et al}.
\newblock {Qiskit: An Open-source Framework for Quantum Computing}, 2021.

\bibitem{griffiths1999introduction}
David~J. Griffiths.
\newblock Introduction to electrodynamics, 1999.

\bibitem{chung2010computational}
T.J. Chung.
\newblock {\em Computational fluid dynamics}.
\newblock Cambridge University Press, 2010.

\bibitem{blazek2015computational}
Jiri Blazek.
\newblock {\em Computational fluid dynamics: principles and applications}.
\newblock Butterworth-Heinemann, 2015.

\bibitem{ruan2021quantum}
Yue Ruan, Xiling Xue, and Yuanxia Shen.
\newblock Quantum image processing: opportunities and challenges.
\newblock {\em Mathematical Problems in Engineering}, 2021, 2021.

\bibitem{cerezo2021cost}
Marco Cerezo, Akira Sone, Tyler Volkoff, Lukasz Cincio, and Patrick~J Coles.
\newblock Cost function dependent barren plateaus in shallow parametrized
  quantum circuits.
\newblock {\em Nature communications}, 12(1):1--12, 2021.

\bibitem{ma1994structural}
Zheng-Dong Ma, Hsien-Chie Cheng, and Noboru Kikuchi.
\newblock Structural design for obtaining desired eigenfrequencies by using the
  topology and shape optimization method.
\newblock {\em Computing Systems in Engineering}, 5(1):77--89, 1994.

\end{thebibliography}

\end{document}